\documentclass[10pt, journal, twocolumn]{IEEEtran}

\usepackage{epsfig,graphics,amssymb}
\usepackage{epsfig,graphics,amssymb,bbm}
\usepackage{amscd}
\usepackage{amsmath}
\usepackage{dsfont}
\usepackage{booktabs}
\usepackage{tabularx}
\usepackage{xcolor}
\usepackage{mathtools}
\usepackage[nolist]{acronym}
\usepackage{amscd}
\usepackage{amsmath}
\usepackage[font=footnotesize]{caption}
\usepackage{enumerate}
\usepackage{theorem}
\usepackage{cite}
\usepackage{tikz,pgfplots}
\usepackage{algorithm}
\usepackage{algpseudocode}
\usepackage{subcaption}

\newcommand{\setP}{\mathbbmss{P}}
\newcommand{\setR}{\mathbbmss{R}}

\newcommand{\setC}{\mathbbmss{C}}

\newcommand{\Ex}[1]{\mathbbm{E} \left\lbrace #1 \right\rbrace}

\newcommand{\maP}{\mathcal{P}}

\newcommand{\bxx}{\mathbf{x}}
\newcommand{\bss}{\mathbf{s}}
\newcommand{\bvv}{\mathbf{v}}
\newcommand{\bmm}{\mathbf{m}}

\newcommand{\brc}[1]{\left( #1 \right)}

\newcommand{\br}{{\boldsymbol{r}}}
\newcommand{\bu}{{\boldsymbol{u}}}

\newcommand{\bc}{{\mathbf{c}}}

\newcommand{\bgg}{{\mathbf{g}}}

\newcommand{\byy}{{\mathbf{y}}}
\newcommand{\bzz}{{\mathbf{z}}}
\newcommand{\baa}{{\mathbf{a}}}

\newcommand{\trp}{\mathsf{T}}
\newcommand{\her}{\mathsf{H}}

\newcommand{\mA}{\mathbf{A}}

\newcommand{\mR}{\mathbf{R}}

\newcommand{\mI}{\mathbf{I}}

\newcommand{\mJ}{\mathbf{J}}

\newcommand{\mW}{\mathbf{W}}
\newcommand{\mC}{\mathbf{C}}

\newcommand{\mM}{\mathbf{M}}

\newcommand{\mH}{\mathbf{H}}
\newcommand{\mZ}{\mathbf{Z}}

\newcommand{\mse}{\mathrm{MSE}}
\newcommand{\crb}{\mathrm{CRB}}

\newcommand{\norm}[1]{\left\lVert #1 \right\rVert}
\newcommand{\abs}[1]{\left\lvert #1 \right\rvert}
\newcommand{\tr}[1]{\mathrm{tr} \{ #1 \}}

\newtheorem{definition}{Definition}

\newtheorem{theorem}{Theorem}

\newtheorem{lemma}{Lemma}
\newtheorem{remark}{Remark}
\IEEEoverridecommandlockouts
\usepackage{titlesec}
\titlespacing*{\section}{0pt}{4pt}{4pt}
\titlespacing*{\subsection}{0pt}{4pt}{4pt}
\begin{document}

\begin{acronym}
	\acro{mimo}[MIMO]{multiple-input multiple-output}
    \acro{mse}[MSE]{mean square error}
    \acro{ssl}[SSL]{Sensing signaling load}
	\acro{mmwave}[mmW]{millimeter wave}
	\acro{tdd}[TDD]{time division duplexing}
	\acro{sinr}[SINR]{signal-to-interference-plus-noise ratio}
	\acro{csi}[CSI]{channel state information}
	\acro{ao}[AO]{alternating optimization}
	\acro{rhs}[r.h.s.]{right hand side}
	\acro{lhs}[l.h.s.]{left hand side}
	\acro{awgn}[AWGN]{additive white Gaussian noise}
	\acro{iid}[i.i.d.]{independent and identically distributed}
	\acro{isac}[ISAC]{integrated sensing and communication}
	\acro{tas}[TAS]{transmit antenna selection}
	\acro{rf}[RF]{radio frequency}
	\acro{srzf}[SRZF]{secure regularized zero-forcing}
	\acro{irs}[IRS]{intelligent reflecting surface}
	\acro{mm}[MM]{majorization-maximization}
	\acro{mmse}[MMSE]{minimum mean square error}
 \acro{sgd}[SGD]{stochastic gradient descent}
  \acro{svd}[SVD]{singular value decomposition}
	\acro{fp}[FP]{fractional programming}
	\acro{bs}[BS]{base station}
	\acro{bcd}[BCD]{block coordinate descent}
	\acro{qos}[QoS]{quality-of-service}
 \acro{moop}[MOOP]{multi-objective optimization}
  \acro{soop}[SOOP]{single-objective optimization}
	\acro{fl}[FEEL]{federated edge learning}
    \acro{ota-fl}[OTA-FEEL]{over-the-air federated edge learning}
	\acro{los}[LoS]{line-of-sight}
	\acro{snr}[SNR]{signal-to-noise ratio}
	\acro{ps}[PS]{parameter server}
	\acro{kkt}{Karush-Kuhn-Tucker}
	\acro{bcd}[BCD]{ block coordinate descent}
 \acro{sdr}[SDR]{semidefinite relaxation}
\acro{ml}[ML]{ maximum-likelihood}
\acro{crb}[CRB]{Cram\'er-Rao bound}
\end{acronym}
\IEEEoverridecommandlockouts
\title{Multi-objective  Optimization for Over-the-Air Federated Edge Learning-enabled Collaborative Integrated Sensing and Communications}
\author{
		\IEEEauthorblockN{
			Saba Asaad, 
            Hina Tabassum, \textit{Senior Member,~IEEE}
            and Ping Wang, \textit{Fellow,~IEEE}
			\thanks{Saba Asaad, Hina Tabassum and Ping Wang are with the Department of Electrical Engineering and Computer Science at York University, Toronto, Canada; emails:  \{asaads, hinat, pingw\}@yorku.ca.}
		}
	}
\maketitle
\raggedbottom
	
\begin{abstract}
    This paper introduces a novel multi-objective  integrated sensing and communications (ISAC) framework to enable collaborative wireless sensing in conjunction with over-the-air federated-edge learning (OTA-FEEL). The framework enables multi-task OTA aggregation to handle sensing and learning simultaneously, while benefiting from dual-purpose uplink signals for both communications and target sensing.
    Starting from characterizing the local sufficient statistics at each edge device and establishing its stationarity, we develop a tractable analytical expression for the local sufficient statistics. To suppress the interference  from uplink transmissions of other devices through matched filtering, we then propose a novel orthogonal pulse shaping method. 
    Then, we derive optimal unbiased estimate of the target's coordinates by casting the centralized problem of joint likelihood function maximization of all devices as the distributed likelihood maximization of each device (which requires only local sufficient statistics). A lower bound on the sensing error variance is then characterized using the Cramér-Rao bound (CRB). We then formulate multi-objective optimization (MOOP) problem  to minimize the \ac{mse} and sensing error bound simultaneously. The considered problem is then solved using $\epsilon$-constrain method. Numerical results demonstrate that the proposed dual-purpose OTA-FEEL-enabled collaborative ISAC framework enhances sensing accuracy without adversely affecting the performance of the primary OTA-FEEL task. While conventional single-shot collaborative sensing schemes are limited by the average error of local  estimators, the proposed algorithm achieves the CRB of the considered  problem. 
\end{abstract}
\IEEEoverridecommandlockouts
\section{Introduction}
With the continuous advancements in wireless communications, the need for systems capable of performing multiple tasks using shared resources is becoming increasingly critical. 
For instance, \ac{isac}  enables communication and sensing tasks within the same platform. 
\Ac{isac} is designed to efficiently share spectrum, hardware, and signal processing resources between communication and radar-like sensing applications \cite{liu2022survey}. This dual functionality enables enhanced sensing capabilities at minimal cost while opening up frequency resources. Researchers have explored numerous techniques for the joint design of communication and sensing systems, such as waveform design, joint resource allocation and advanced signal processing algorithms \cite{zhao2022radio}. 
These approaches enable simultaneous transmission of communication and sensing data, requiring advanced signal processing across multiple aspects of the \ac{isac} system. 

Initial research studies in \ac{isac} typically considered centralized sensing where a centralized unit, e.g., a \ac{bs}, is responsible for both sensing and communications. For instance, in \cite{liu2018mu}, \ac{mimo}-\ac{isac} beamforming has been designed for shared and separate radar-communication deployments  with sensing at the BS to detect the target. In \cite{liu2019hybrid}, mmWave \ac{mimo}-\ac{isac} with hybrid beamforming was analyzed, framing the hybrid beamformer design as an optimization problem that balances communication and sensing performance  performed at the \ac{bs}. 
\color{black}
While centralized approaches can be effective under ideal conditions—when the central unit has sufficient computational and communication resources—they often become impractical in real-world deployments due to high signaling overhead, limited scalability, and elevated energy consumption \cite{eisele2023multiple}.
For instance, in cloud-RAN ISAC systems, centralized sensing imposes substantial signaling load due to fronthaul transmission of raw signals. As shown in \cite{zou2024distributed}, distributed sensing can reduce this load by up to $70\%$ under realistic system parameters. Similarly, \cite{xu2016distributed} demonstrates that distributed estimation in UAV networks eliminates the need for long-range communication with a central controller, significantly reducing energy consumption while maintaining estimation accuracy close to that of centralized methods. In asynchronous sensor networks, \cite{yuan2018expectation} shows that distributed localization reduces communication overhead by a factor of 10, achieves comparable accuracy, converges rapidly, and offers greater resilience to link failures.
These results highlight the limitations of centralized sensing in bandwidth and energy-constrained environments, motivating the need for scalable distributed solutions.\footnote{\color{black} See Section~V.C for further quantitative discussions in this respect.}
\color{black}

To address these challenges, we propose an \textit{\ac{ota-fl}-enabled collaborative ISAC framework} in which target estimation is performed in a distributed manner across multiple network nodes. 
The proposed framework enables  sensing objects in a collaborative manner, while using a single signal optimized for both sensing and communications.  That is, when the users transmit local parameters of sensing and computation tasks to the \ac{ps}, they simultaneously receive reflections from nearby objects. These reflections are then utilized for local estimation, and the resulting estimates are further improved using \ac{ota-fl} framework. This approach enables reducing communication overhead and resource consumption.


\subsection{Related Work}
Distributed sensing offers promising benefits in terms of scalability and reduced reliance on centralized infrastructure, yet practical realizations remain limited in the literature.
\color{black}A closely related study is \cite{li2023integrated}, which proposed a decentralized approach for target localization, where each sensor locally estimates the target’s position using reflected radar signals and sends it over-the-air in a single round for averaging at the server. This one-shot method focuses solely on sensing, without federated learning or multiple communication rounds.
\color{black}
Subsequently, designing a fully distributed sensing approach with performance levels comparable to centralized methods is crucial. In this context, recent research has begun to explore the integration of \textit{centralized} \textit{sensing into \ac{ota-fl} }systems.  
\color{black}  In \cite{du2024integrated}, the authors focused on executing a federated learning task alongside radar sensing within an \ac{isac} framework. \color{black} In this system, the \ac{ps} efficiently transmits integrated signals to communicate with edge devices while performing target detection based on the reflections of these signals. This approach highlights the ISAC concept by enabling the PS to perform sensing and communication simultaneously, but without involving edge devices in sensing. \color{black}
Unlike  \cite{du2024integrated} where centralized sensing was considered, \cite{liu2022toward} explores standard \ac{fl} with the objective of loss function minimization  while considering time-switching \ac{isac} waveform. The waveform allows each device to perform sensing and communication in an orthogonal manner. 
This work assumes interference-free transmissions due to sufficient device separation, which may not hold in practical settings.
In \cite{zheng2024over}, a joint client selection and power allocation problem was studied for centralized sensing in an \ac{ota-fl} system, where all clients send their sufficient statistics to the central server for joint detection via hypothesis testing. A flexible client selection scheme was proposed, allowing devices to either join FEEL training or only send sensing data. The study in \cite{liu2022toward} considered \textit{collaborative sensing solely}, and \cite{zheng2024over} separates sensing from learning tasks. Our study, however, focuses on proposing a dual-purpose ISAC signaling design that simultaneously executes distributed sensing and data transmission using a common signaling method. In dual-purpose ISAC design, sensing data is collected through a common signal. This approach is distinct from sensing-specific designs, where signals are pre-designed to satisfy favorable properties, such as orthogonality. These pre-designed signals are unsuitable for dual-purpose signaling because the transmitted signal is primarily optimized to transfer data. This leads to strong undesired interference that degrades sensing performance.


\subsection{Contributions}
\textcolor{black}{Joint sensing and FEEL in wireless systems can be considered in two settings: ($i$) sensing coexists with FEEL (i.e., FEEL does not assist sensing), e.g., \cite{asaad2025FL_ISAC}, and ($ii$) \ac{fl} is used solely for sensing, i.e., users collect data, train local models, and then perform global model aggregation.
The former reflects centralized sensing, while the latter represents a distributed setting that can potentially benefit from integration with distributed learning.}  In this paper, we develop a novel multi-objective ISAC framework, referred to as \textit{CollabSenseFed}, to enable OTA-FEEL-empowered collaborative wireless sensing. The framework utilizes multi-task OTA aggregation to handle both sensing and learning, while benefiting from dual-purpose uplink signals for both communications and sensing. \footnote{\color{black} We use standard Federated Averaging for model aggregation, given its effectiveness in synchronous networks. Alternative methods, e.g.,\cite{wang2025communication} may enhance performance and can be incorporated into our framework in future.\color{black}}
The key contributions of this paper can be summarized as:

$\bullet$ {\bf Multi-task OTA Aggregation:} 
 We develop a framework for multi-task OTA computation, in which uplink-downlink array beamforming is optimized to enable simultaneous computation and aggregation of multiple tasks. This extension to OTA computation  enables us to process multiple models  over-the-air in parallel, thus allowing for gradient aggregation for the localization task and the \ac{fl} model aggregation. 
 
 $\bullet$ {\bf Dual-Purpose Signaling:}  {We propose a dual-purpose ISAC signaling in which the reflections from uplink transmit signals are used to sample the target's position. The dual-signaling prevents us from using a pre-defined orthogonal sensing signals. This causes interference on the received reflection of each device from reflections of the other devices' information-bearing signals, leading to very noisy sensing samples. To address this issue, we propose a novel pulse shaping in which the devices modulate their uplink transmission by predefined orthogonal pulses. }
 This approach enables devices to effectively suppress interference from uplink transmissions of other devices while preserving OTA model aggregation. 
 
$\bullet$ {\bf Collaborative sensing via OTA-FEEL:} 
Starting from characterizing the local sufficient statistics at each device, we establish its stationarity to simplify the local sufficient statistics. Then, we derive optimal unbiased estimate of the target's coordinates by casting the centralized problem of joint likelihood function maximization of all devices as the distributed likelihood maximization of each device. To further quantify the sensing error, we invoke the \ac{crb}, which determines a lower bound on the variance the target location estimate. This bound is given in terms of the inverse of the Fisher information matrix, which cannot be derived analytically. We leverage its basic properties to obtain an analytic bound to analytically quantify the sensing error.  
 
$\bullet$ {\bf Multi-objective Optimization:} We formulate a \ac{moop} to minimize the aggregation error and sensing error bound with respect to uplink and downlink beamformers. We invoke the $\epsilon$-constraint method to transform the \ac{moop} into a constrained \ac{soop}. We then utilize a \ac{bcd} scheme to approximate the optimal beamformers by alternating between two marginal problems: one that minimizes the aggregation error while constraining the sensing error, and the other that optimizes the sensing metric subject to power constraints. We approximate the latter problem with a relaxed convex optimization whose solution can be readily computed via convex programming. Using this result, we develop a low-complexity algorithm that approximates the optimal beamformers iteratively. 
 
$\bullet$ We validate the proposed scheme through extensive numerical experiments on MNIST and CIFAR-10. Results show that CollabSenseFed’s localization error closely tracks the CRB for joint estimation, offering a notable improvement over conventional methods, where the error is bounded by the average of individual estimators.

	
\textit{Notations}: 
Scalars, vectors, and matrices are denoted by non-bold, bold lowercase, and bold uppercase letters, respectively. For any matrix $\mH$, $\mH^{\trp}$, $\mH^*$, and $\mH^{\her}$ denote the transpose, conjugate, and Hermitian transpose. $\mI_N$ is the $N \times N$ identity matrix; $\setR$ and $\setC$ denote the real and complex sets. $\mathcal{CN}(\eta, \sigma^2)$ is the complex Gaussian distribution with mean $\eta$ and variance $\sigma^2$. $\dagger$ denotes the pseudo-inverse, $\otimes$ the Kronecker product, and $\Ex{.}{}$ the expectation. We write ${1,\ldots,N}$ as $[N]$, $\nabla$ for gradient, and $\norm{\cdot}_{\rm F}$ for Frobenius norm.

\section{System Model And Assumptions}
We consider a \ac{fl} setting in which a single \ac{ps} and $K$ devices are training a global model over a wireless network. We assume that the devices and the \ac{ps} are equipped with $M$ and $N$-antennas, respectively. In addition to the FEEL training, the devices share sensing parameters, such as {gradients of local localization losses}, with the \ac{ps}. At this stage, we focus on the over-the-air computation aspect of the system, while the nature of these sensing parameters will be discussed in Section \ref{sec:Localize}. 
We assume that the devices perform $I$ simultaneous computation tasks over-the-air. Setting $I=1$ reduces the problem to the classical single-function over-the-air computation, while $I=2$ models the scenario where both \ac{fl} and localization-related computations are carried out simultaneously. We refer to this problem as \textit{multi-function or multi-task computation} and formulate it precisely in the presence of a reflecting target.
\begin{figure}[t]
 		\centering
		\hspace*{-.2cm}
		\includegraphics[scale=.55]{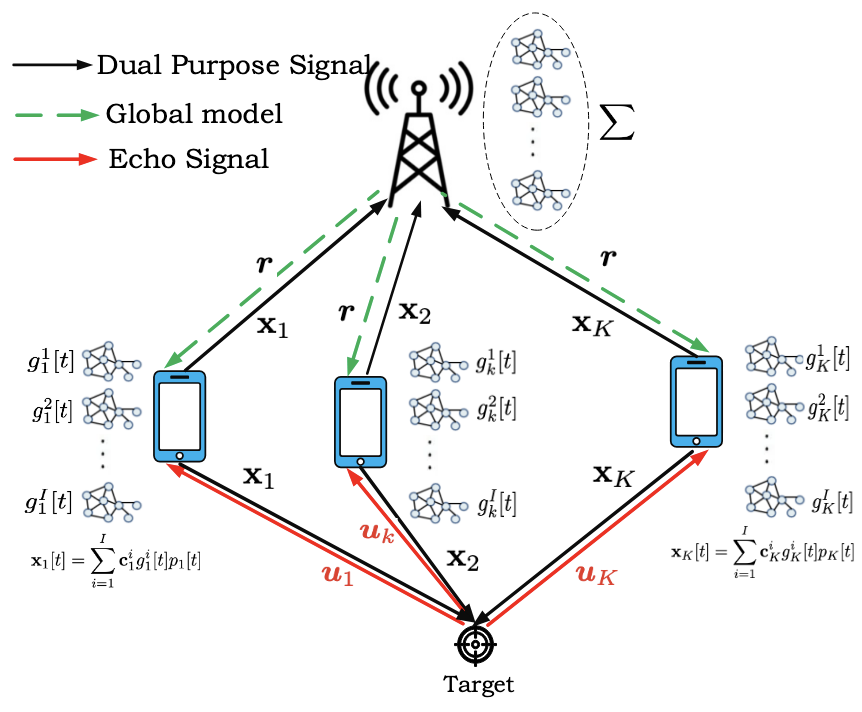}
		\caption{ Illustration of the multi-task OTA-FEEL set-up.}
        \vspace{-.51cm}
		\label{Fig1_Schematic}
\end{figure}
  \vspace{-.51cm}
\subsection{Multi-function (or Multi-task) Computation Model}
The devices aim to transmit their parameters for the $I$ different tasks simultaneously. These parameters may include model updates for the \ac{fl} task or gradients of localization losses for the sensing task. We assume that these parameters are computed locally at the devices and then shared with the \ac{ps}. Let $g_k^i[t]$ denote the parameter of task $i\in [I]$ computed locally by device $k$ at symbol interval $t$, e.g., a model parameter in \ac{fl} problem or an entry of localization loss's gradient. Device $k$ then computes its \textcolor{black}{analog} transmit symbol at time $t$ from its $I$ parameters $g_k^1[t], \ldots, g_k^I[t]$ as $\bss_k [t] = \sum_{i=1}^I \bc_k^{i} g_k^i [t]$ for some $\bc_k^i\in\setC^M$ that precodes the parameter $g_k^i[t]$. This can be written as follows: 
\begin{align}\label{eq:trS}
	\bss_k [t] = \mC_k \bgg_k [t],
\end{align}
where $\mC_k = [\bc_k^1, \ldots,\bc_k^I ]$ and $\bgg_k [t] = [g_k^1 [t], \ldots,g_k^I[t] ]^\trp$. 
We assume that for all $k$ and $i$, $g_k^i[t]$ describes a stationary process with auto-correlation $R_k [\tau]$ defined as $R_k[\tau] = \Ex{g_k^i [t]g_k^i [t-\tau]^*}$ that decays against $\tau$, i.e.,  $R_k[\tau] \leq R_k[0]$ for all $\tau$.
In practice, parameter updates across various tasks often show similar temporal correlations, thus we consider the auto-correlation is the same across tasks (i.e., the same for all $i$) \cite{li2020federated}.
The mean and variance of these processes are assumed to be zero and one, respectively, i.e., 
$\Ex{g_k^i [t]} = 0$ and $\Ex{g_k^i [t]g_k^i [t]^*} = R_k[0] = 1$.\footnote{This is consistent with the fact that parameters are typically centered and scaled before transmission in distributed learning \cite{sedaghat2023novel, hamidi2025over}.}
We further assume that the computations are performed over independent functions, e.g., independent mini-batches of data, and hence the parameters $g_k^i [t]$ are independent across devices and tasks, i.e., 
$\Ex{g_k^i [t]g_\ell^j [t]^*}= 0$ for all $k\neq \ell$ and $i\neq j$. The devices are limited in transmit power with upper bound value $P$, i.e., $ \Ex{\norm{\bss_k [t]}^2} \leq P$. Considering the statistical model of the local parameters, this concludes that the scaling factors satisfy $ \sum_{i=1}^I \norm{\bc_k^{i}}^2 = \norm{\mC_k}_{\rm F}^2 \leq P$ for $k\in[K]$.

\subsection{Dual-Purpose ISAC Signal Model}
We consider a stationary target whose reflections are received by the devices after a delay that depends on the distance of each device from the target. \textcolor{black}{We assume that the target is located within the sight of the devices, i.e., the reflections are received only by the devices.} The devices aim to utilize the target's reflections within a fixed time-frame to jointly localize the target's position. To this end, they construct a dual-purpose transmit signal from the symbol $\bss_k[t]$ and a predefined pulse, which can be used for both OTA computation and localization.

We focus on a time-frame consisting of $T$ symbol intervals. Device $k$ constructs its transmit signal by multiplying its transmit symbols by a predefined pulse $p_k[t]$ of length $T$.  This means that it sends in interval $t \in [T]$ the signal
\begin{align}
	\bxx_k[t] = \bss_k[t] p_k[t].
\end{align}
The pulses $p_k[t]$ for $k\in [K]$ build an orthogonal basis, i.e., for any $k\neq \ell$, we have $ \sum_{t=1}^T p^*_k[t] p_\ell [t] = 0$.
As we show later, this enables the devices to cancel out the interference on their received reflections \footnote{\color{black}Although orthogonal pulse shaping theoretically eliminates inter-device interference as $T \to \infty$, practical implementations with finite $T$ may experience residual interference. This can perturb $\tilde{\Xi}_k$, degrade the quality of sufficient statistics, and bias the estimate of $v^\star$. However, these effects can be mitigated by longer pulse durations, improved pulse design, or interference-aware estimation techniques.}.\color{black} To keep the pulse not impacting the transmit power, we assume $ \abs{p_k[t]}^2 =1$ for $t\in[T]$. This simply implies that $\frac{1}{T}	\sum_{t=1}^T \abs{p_k[t]}^2 = 1$. Examples of such orthogonal pulses are given by discrete Fourier Transform (DFT) and Walsh-Hadamard matrices \cite{seberry2005some}. \textcolor{black}{Note that due to the use of over-the-air computation, the devices are synchronized. As a result, their pulses can be designed to be perfectly orthogonal. It is assumed that the pulses $p_k[t]$ for $k\in [K]$ are predefined and known to all the devices and \ac{ps}. As shown in Fig. \ref{Fig1_Schematic}, the devices send their transmit signals over the wireless channel to the \ac{ps}. These signals serve two purposes: 1) they superimpose over-the-air and are used by the \ac{ps} to compute the aggregated parameters, and 2) they provide the devices with reflections from the target, which are used to localize the target.}


\subsection{Over-the-Air Computation Model}
At interval $t$, the \ac{ps} receives a superimposed version of the transmitted signals, i.e., 
\begin{align} 
	\hspace{-.35cm}\byy[t] = \sum_{k=1}^{K} \mH_k \bxx_k[t] + \bzz [t]
 = \sum_{k=1}^{K} \mH_k \bss_k[t] p_k[t] + \bzz [t],\label{eq:yPS}
\end{align}
where $\bzz [t]\sim\mathcal{CN}\brc{0, \sigma^2\mI_N}$ is the additive white Gaussian noise at the PS, and $\mH_k\in\setC^{N\times M}$ is the matrix of channel coefficients between device $k$ and the \ac{ps} which is assumed to be constant within the time-frame, i.e., the channel coherence time is larger than $T$. By plugging \eqref{eq:trS} into \eqref{eq:yPS}, we can see that the received signal at the \ac{ps} contains the superimposed version of all local parameters, i.e., 
\begin{subequations}
    \begin{align}
	\byy[t]  = \sum_{k=1}^{K}  \mH_k \mC_k \bgg_k [t]  p_k[t]+ \bzz [t].
\end{align}
\end{subequations}

The \ac{ps} intends to compute a predefined aggregation on the local parameters for each function \textcolor{black}{directly over the air, i.e., in the analog domain}. This means that for each $i = 1,\ldots, I$, it desires to compute $\bar{g}^i [t] = \sum_{k=1}^K w_k^i g_k^i [t]$ for some averaging weights $w_k^i$ with $i=1,\ldots, I$ and $k=1,\ldots, K$. To this end, it uses $I$ separate linear receivers $\bmm_1 [t], \ldots, \bmm_I[t]$ to compute $I$ different post-processing components from the received signal as $r^j [t] =  \bmm_j^\her[t]  \byy[t]$. By defining the receive beamforming matrix $\mM[t] = [\bmm_1 [t], \ldots, \bmm_I[t]]$ and, the vector of aggregated parameters as $ \br [t] = [r^1 [t], \ldots, r^I [t] ]^\trp$, we can express the aggregated model parameter compactly as follows:
\begin{subequations}
\begin{align}
	\br [t] &=  \mM^\her [t]  \byy[t],\\&=  \sum_{k=1}^{K}  \mM^\her [t] \mH_k \mC_k \bgg_k [t]  p_k[t]+ \mM^\her[t] \bzz [t].
 \label{eq:RR}
\end{align}
\end{subequations}


From \eqref{eq:RR} it can be seen that the \ac{ps} through transmit and receive beamforming can recover a distorted version of the desired aggregations. This describes the concept of over-the-air computation. The ultimate goal in this setting is to find beamformers for the \ac{ps} and devices such that $r^j[t]$ closely approximates the desired $\bar{g}^j[t], \forall j$. Due to channel noise and task/device interference, this goal is only achieved with some error, which we refer to as the \textit{aggregation error}. At time $t$, the aggregation sum-error is given by \eqref{eq:mse} at the top of the next page. In this expression, we define $\bar{\bgg} [t] = [\bar{g}^1 [t], \ldots,\bar{g}^I[t] ]^\trp$, and $\mW_k\in\setR^{I\times I}$ is a diagonal matrix whose diagonal entries are $w_k^1,\ldots,w_k^I$. 
  \begin{figure*}
\begin{align}\label{eq:mse}
	&\mse\brc{\mM [t], \{\mC_k\}} = \Ex{ \norm{\br [t]- \bar{\bgg} [t]}^2 }= \Ex{ \norm{\sum_{k=1}^{K}  \mM^\her [t] \mH_k \mC_k \bgg_k [t]  p_k[t]+ \mM^\her[t] \bzz [t]- \sum_{k=1}^K \mW_k \bgg_k [t]}^2 },\nonumber\\[-5pt]
&\hspace{-.4cm}=\Ex{ \norm{\sum_{k=1}^{K}  \brc{\mM^\her [t] \mH_k \mC_k   p_k[t] - \mW_k } \bgg_k [t]+ \mM^\her[t] \bzz [t] }^2}= \sum_{k=1}^{K}  
\norm{ {\mM^\her [t] \mH_k \mC_k   p_k[t] - \mW_k } }_F^2
+ \sigma^2 \norm{\mM [t]}_{\rm F}^2.
\end{align}
 \hrule
\end{figure*}
The main goal of the \ac{ps} and the devices is to minimize the time average of the aggregation sum-error, i.e.,
\begin{align}
	\hspace{-.31cm}\overline{\mse} &\brc{\{ \mM [t] \}, \{\mC_k\}} = \frac{1}{T} \sum_{t=1}^T  \mse \brc{\mM[t], \{\mC_k \}}.\label{eq:7}
\end{align}
In the following, we take this metric as a measure that quantifies the quality of over-the-air computation. \color{black}  Note that the MSE model in \eqref{eq:7} captures the impact of both communication and sensing.
The sensing-related gradients contribute to the transmitted signal and introduce an additional source of interference, which propagates through the aggregation function. Thus, the derived MSE expression inherently accounts for the impact of sensing by incorporating the errors arising from the estimated parameters within the computation process. To address this, the precoding matrix $\mathbf{C}_k = [\bc_k^1, \ldots,\bc_k^I ]$ is designed to balance both computation and sensing tasks. This trade-off is explicitly addressed in the \ac{moop} in Section IV.
\color{black}

\subsection{Received Echo Signal Model}
Given the radiated signals, each device receives a superposition of echoes from the target, including the~echoes of all transmitted signals. These reflections can be used to jointly estimate the location of the target. Let the coordinates of the target be denoted by $\bold{v} = [\ell_x, \ell_y, \ell_z]$. The echo received from the target at device $k$ at time $t$ is given by
\begin{align}
	\bu_k[t]  = \sum_{\ell=1}^K \baa_{k} \brc{\bold{v}} \baa_{\ell} \brc{\bold{v}}^\trp \bxx_\ell [t]  + \boldsymbol{\nu}_k [t],\label{eq:u}
\end{align}
where $\baa_{k}\brc{\bold{v}}$ represents the array response of device $k$ at the angle-of-arrival corresponding to the line connecting device $k$ to the target at position $\bold{v}$. 
The term $\boldsymbol{\nu}_k[t] \in \setC^M$ further denotes additive white Gaussian noise with mean zero and variance $\varsigma_k^2$. 
Using the far-field propagation model \cite{proakis2007digital}, the array response of  device $k$ is then given by:
\begin{align}\label{array_resp}
	\hspace{-.2cm}\baa_{k}\brc{\bold{v}} = \alpha_k [ 1, e^{-j\frac{2\pi d}{\lambda} \sin \theta_k\brc{\bold{v}} }, \ldots, e^{-j\frac{2\pi \brc{M-1} d}{\lambda} \sin \theta_k\brc{\bold{v}}}],
\end{align}
where $\alpha_k=\frac{\alpha_0}{\norm{\bvv-\bvv_k}^2}$ with $\alpha_0$ being a constant that represents the reference path-loss. $\bold{v}_k= [\ell_x^k,\ell_y^k,\ell_z^k]$ is the location of device $k$, $\theta_k(\bold{v})$ denotes the angle-of-arrival by the line connecting $\bold{v}$ to $\bold{v}_k$. 
This relationship satisfies $\sin \theta_k\brc{\bold{v}} = \frac{\ell_z - \ell_z^k }{\Vert \bold{v} - \bold{v}_k \Vert},$ where $\ell_z$ and $\ell_z^k$ are  $z$-coordinates of the target and device $k$, respectively. Thus, the angle is influenced by their height difference in 3D space. $\lambda$ is the wavelength of the transmitted signal, and $d = \Vert \bold{v} - \bold{v}_k \Vert$ is the distance between the target and device $k$. 
Given the reflections $\bu_k[1], \ldots, \bu_k[T]$, devices aim to localize target by estimating $\mathbf{v}$. They filter the received signals to compute sufficient statistics, which are then used to derive an estimator for $\mathbf{v}$. 

\section{OTA-FEEL-enabled Collaborative Sensing}
The ultimate goal of the devices is to address both sensing problem, i.e., estimating $\mathbf{v}$ from $\bu_k[1], \ldots, \bu_k[T]$, and the aggregation simultaneously. To this end, upon receiving the reflections, they design uplink signals that contain some information regarding both sensing and learning tasks. They then transmit their uplink signals and receive the next set of reflections. To accomplish this joint task efficiently, there are two design problems to be addressed: (1) designing precoders $\mC_k$ and beamformers $\bmm_i[t]$ to enhance both learning and sensing performance, and (2) eveloping a joint algorithmic approach. We tackle these through the following steps:
\begin{itemize}
\item We first specify how the device $k$, upon receiving $\bu_k[1], \ldots, \bu_k[T]$, extracts sufficient statistics ${\Xi}_k$. 
\item We study properties of the sufficient statistics in \textbf{Lemmas~\ref{lem:1} and \ref{lem:2}}, which help us derive simplified  sufficient statistics denoted by $\hat{\Xi}_k$.
\item Given $\hat{\Xi}_k$, we determine the optimal unbiased estimate of the target coordinate $\bvv$ by maximizing joint likelihood $p\brc{\hat{\Xi}_1 , \ldots, \hat{\Xi}_K  \vert \bvv }$. We show that this estimator reduces to optimizing the sum of local log-likelihoods, where the likelihood of device $k$ depends only on its local sufficient statistics $\hat{\Xi}_k$. Leveraging this, we frame the localization problem as a distributed learning task.
 \item We find an analytic metric for sensing performance by deriving a lower bound on the sensing error variance. Specifically, \textbf{Theorem~\ref{th:1}} provides an explicit expression for the Fisher information matrix in the collaborative sensing setup. Invoking the Cram\'er-Rao theorem, we use the Fisher information to bound the sensing error when the clients solve the maximum-likelihood estimation problem \textit{jointly}.Although the CRB is explicit, it lacks a closed-form in terms of design parameters (e.g., precoders $\mC_k$ for $k \in [K]$). To address this, we use basic properties of the Fisher information matrix to obtain an analytic bound in \textbf{Theorem~\ref{thm:3}}, which serves as a performance metric for collaborative sensing.
 \item We formulate \ac{moop} problem in which both the learning metric, i.e., $\mse\brc{\mM [t], \{\mC_k\}}$, and sensing metric, i.e., lower bound on sensing error $\crb_L \brc{\{\mC_k\}}$ determined in \textbf{Theorem~\ref{thm:3}}, are to be minimized simultaneously. 
\end{itemize}
\vspace{-.1cm}
\subsection{Sufficient Statistics at each Device}
\vspace{-.1cm}
The sensing task of the devices is to use their received reflections over an interval of length $T$ to estimate the target location. To this end,  upon receiving $\bu_k[1], \ldots, \bu_k[T]$, device $k$ extracts sufficient statistics by matching its received sequence  with its transmit signal $\bxx_k[t]$, i.e., it computes
\begin{align}\label{eq:eqeta}
	&\Xi_k \brc{\bold{v}} = \frac{1}{{P }T} \sum_{t=1}^T \bu_k [t] \bxx_k^\her [t],  \nonumber\\
	&= \frac{1}{{P }T} \!\sum_{t=1}^T  \sum_{\ell=1}^K 
 \baa_{k}\!\brc{\bold{v}} \baa_{\ell}\!\brc{\bold{v}}^\trp \bxx_\ell [t] \bxx_k^\her [t]\nonumber
 + \frac{1}{{P}T}  \sum_{t=1}^T \boldsymbol{\nu}_k [t] \bxx_k^\her [t], \nonumber \\&= 
 \baa_{k}\!\brc{\bold{v}}\!\baa_k\!\brc{\bold{v}}^\trp\! \mR_{k,k}\!+\!\hspace{-.3cm}\sum_{\ell=1, \ell\neq k}^K \hspace{-.3cm}\baa_{k}\!\brc{\bold{v}} \baa_\ell\! \brc{\bold{v}}^\trp\!\mR_{\ell,k}
 \!+\!\sum_{t=1}^T\!\frac{\boldsymbol{\nu}_k [t] \bxx_k^\her [t]}{P T}, \nonumber
\end{align}
where the sample correlation matrix $\mR_{k,\ell}$ is defined as $\mR_{\ell,k} = \frac{1}{{P}T}\sum_{t=1}^T  \bxx_\ell [t] \bxx_k^\her [t].$ Due to the orthogonality of the transmit pulses, the match filtering suppresses the interference of other devices at device $k$. Precisely, with large choices of $T$, $\mR_{\ell,k} \to \boldsymbol{0}$ for $\ell \neq k$ as $T\to \infty$. This is shown in the following lemmas.
\begin{lemma}\label{lem:1}
	Let $\bss_k[t]$ be a zero-mean stationary process with auto-correlation matrix $ \tilde{\mR}_k [\tau] = \Ex{\bss_k[t] \bss_k^\her[t-\tau]}$ whose trace decays against $\tau$. As $T \rightarrow \infty$, for any $\ell\neq k$, we have
 \begin{subequations}
     \begin{align}
		\mR_{k,k} &\rightarrow \frac{1}{P} \tilde{\mR}_k[0], \\
		\mR_{\ell,k} &\rightarrow \boldsymbol{0}. \label{eq:eq2}
	\end{align}
 \end{subequations}
\end{lemma}
\begin{proof}
See \textbf{Appendix A}.
\end{proof}
Given the definition of $\bss_k[t]$, it is straightforward to show that its covariance matrix satisfies the decaying constraint stated in \textbf{Lemma~\ref{lem:1}}.  This is shown in \textbf{Lemma~\ref{lem:2}}. 

\begin{lemma}\label{lem:2}
	The random process $\bss_k[t]$ is a zero-mean stationary process with auto-correlation matrix $\tilde{\mR}_k[\tau] =  R_k[\tau] \mC_k\mC_k^\her$ with $R_k[\tau]$ denoting the common auto-correlation of local processes $g_k^i[t]$.
\end{lemma}
\vspace{-.2cm}
\begin{proof}
See \textbf{Appendix B}.
\end{proof}

We can now use \textbf{Lemma~\ref{lem:1}} to simplify the expression for sufficient statistics. We assume that $T$ is large enough, such that the limits of the law of large number hold. In this case, the sufficient statistics is approximately given by:
\begin{align}
	\hspace{-.2cm}\Xi_k \brc{\bold{v}} &\!\approx\!
 \baa_{k}\!\brc{\bold{v}} \baa_k\!\brc{\bold{v}}^\trp \frac{\mC_k\mC_k^\her}{P}
	\!+\! \frac{\varsigma_k}{\sqrt{PT}}
 \hat{\mZ}_k \brc{\frac{\mC_k\mC_k^\her}{P}}^{1/2}\hspace{-.3cm},
\end{align}
where $\hat{\mZ}_k \in\setC^{M\times M}$ is a zero-mean unit variance Gaussian matrix.
We can further whiten the noise term and compute the alternative sufficient statistics $\hat{\Xi}_k \brc{\bold{v}}$ as follows:
\begin{align}
	\hspace{-.3cm}\hat{\Xi}_k \brc{\bold{v}} = 
 \frac{\sqrt{PT}}{\varsigma_k}
 \Xi_k \brc{\bold{v}}\brc{\frac{\mC_k\mC_k^\her}{P}}^{-1/2}.
 \label{eq:suff_stat1}
\end{align}
Let define $\hat{\mA}_k \brc{\bvv}=\mA_k \brc{\bvv}\brc{\mC_k\mC_k^\her}^{1/2}$ where
 \begin{align}\label{eq:A_k}
    \mA_k \brc{\bvv} = \sqrt{\frac{T}{\varsigma_k^2}} \baa_{k} \brc{\bold{v}} \baa_k \brc{\bold{v}}^\trp,
\end{align}
then $\hat{\Xi}_k \brc{\bold{v}}$ can be rewritten as
\begin{align}
	\hat{\Xi}_k \brc{\bold{v}} = \hat{\mA}_k (\bvv)+ \hat{\mZ}_k,
 \label{eq:suff_stat}
\end{align}
which simplifies the likelihood expression \footnote{\color{black}The array response $\baa_k(\bold{v})$ is expressed as a deterministic function of the target's location through the wave-propagation model. As a result, $\hat{\mA}_k(\bvv)$ is fully specified by the estimation parameters and does not require separate CSI acquisition. Consequently, the estimation of $\hat{\mA}_k(\bvv)$ directly yields an estimate of the target parameters.}.
\color{blue}
\color{black}

 \subsection{Joint ML Estimation for OTA-FEEL-enabled Localization}
\label{sec:Localize}

Given the sufficient statistics $\hat{\Xi}_k $, the optimal unbiased estimate of the target coordinate $\bvv$ is obtained by maximizing the likelihood of location $\bvv$, i.e., $p\brc{\hat{\Xi}_1 \brc{\bold{v}}, \ldots, \hat{\Xi}_K \brc{\bold{v}}\vert \bvv}$. Thus, we can write 
\begin{align}
    \bvv^\star = \arg\max_{\bvv\in\setR^3} p\brc{\hat{\Xi}_1 \brc{\bold{v}}, \ldots, \hat{\Xi}_K \brc{\bold{v}} \vert \bvv },
\end{align}
where $\hat{\Xi}_k \brc{\bold{v}}$ is defined in \eqref{eq:suff_stat}. 
Since $\hat{\mZ}_k$ is \ac{iid} standard Gaussian, the likelihood function can be expressed as:
\begin{subequations}
	\begin{align}
 \hspace{-.3cm}p\brc{\hat{\Xi}_1 , \ldots, \hat{\Xi}_K  \vert \bvv } &= \prod_{k=1}^K p\brc{\hat{\Xi}_k  \vert \bvv }, \\
&\propto \prod_{k=1}^K \exp\left\{ - \norm{\hat{\Xi}_k - \hat{\mA}_k \brc{\bvv} }_F^2 \right \},
	\end{align}
\end{subequations}
 where we drop the argument $\bvv$  and denote $\hat{\Xi}_k\brc{\bvv}$ as $\hat{\Xi}_k$ for sake of brevity. Since the logarithm is an ascending function and that the likelihood is always non-zero, we can rewrite the joint \ac{ml} localization problem as 
\begin{subequations}
\begin{align}\label{eq:distLoc}
    \bvv^\star &= \arg\max_{\bvv\in\setR^3} \log p\brc{\hat{\Xi}_1, \ldots, \hat{\Xi}_K \vert \bvv },\\
    &= \arg\max_{\bvv\in\setR^3} \sum_{k=1}^K - \norm{\hat{\Xi}_k - \hat{\mA}_k \brc{\bvv} }_F^2.
\end{align}
\end{subequations}

Noting that device $k$ only knows its own location, it is evident that \eqref{eq:distLoc} describes a distributed optimization problem.
The distributed localization problem can be solved by invoking the federated averaging algorithm \cite{mcmahan2017communication}, and hence incorporated into the proposed multi-function computation framework. 
To illustrate this, let us denote the local loss of device $k$ in \ac{ml} localization as follows:
\begin{align}
	\ell_k (\bvv) = \norm{\hat{\Xi}_k - \hat{\mA}_k \brc{\bvv} }_F^2.
\end{align}
The optimization in \eqref{eq:distLoc} is equivalent to $\min_{\bvv\in\setR^3} \ell (\bvv)$, where $\ell (\bvv) = \sum_{k=1}^K \ell_k (\bvv)$. 
By interpreting $\ell_k (\bvv)$ and $\ell (\bvv)$ as local and global loss functions, respectively, we can frame this problem as a distributed learning task.  
\vspace{-.35cm}
\color{black}
\begin{remark}
In practical scenarios with finite-length pulses and/or imperfect orthogonality
due to impairments (such as   multipath fading, Carrier Frequency Offset,
hardware imperfections), the inter-device interference is not completely suppressed. This leads to an extra interference term in \eqref{eq:suff_stat1}, reducing the accuracy of the sufficient statistics and potentially introducing a bias and/or increasing the variance in the estimate of $v^\star$. To mitigate such impacts, one may consider using strategies such as long pulse duration, adopting improved pulse design, and interference-aware estimation methods in practice. The discussion on these approaches is out of the scope of this study and is left as a direction for future work.
\end{remark}
\color{black}
\vspace{-.35cm}
\subsection{Deriving the Cram\'er-Rao Bound}\label{crbMSE}
In radar sensing, we focus on target estimation performance using the \ac{crb}, which provides a lower bound on the variance of any unbiased estimator \cite{liu2021cramer} and is a widely used benchmark in radar systems. To this end, let $\hat{\mathbf{v}}$ represent the estimator of target location $\mathbf{v}$. The CRB matrix can then be given as the inverse of the Fisher Information Matrix (FIM)  \cite{kay1993statistical}, i.e., 
\begin{align}\label{crb}
    \crb \brc{\{\mC_k\} \vert \bvv} =\tr{\mJ^{-1}\brc{\bvv}},
\end{align}
with $\mJ\in\setC^{3\times 3}$ representing the FIM\footnote{This framework can be easily extended to cases with more than three estimation variables, accommodating a higher-dimensional $\bvv$ as required by different application scenarios.} defined as follows:
\begin{align}
    	\mJ\brc{\bvv} = -\Ex{ \nabla^2_{\bvv} \log p\brc{\hat{\Xi}_1 , \ldots, \hat{\Xi}_K  \vert \bvv } }.
\end{align}
 Moreover, we denote $\{\mC_k\}$ as an argument in \eqref{crb}, as the reflection signals depend on the precoders $\{\mC_k\}$ through the transmit signals $\{\bxx_k[t] \}$. 
 In general, $\crb\brc{\{\mC_k\} \vert \bvv}$ is a function of unknown $\bvv$, which is the target location. Noting that $\bvv$ is unknown, we consider the worst-case \ac{crb} as the sensing error metric, i.e., the maximum \ac{crb} against all possible choices of unknown $\bvv$. We call this metric worst-case \ac{crb} and denote it as follows:
\begin{align}
	\crb_{\text{worst}} \brc{\{\mC_k\}} = \max_{\bold{v}\in \setP} \crb \brc{\{\mC_k\} \vert \bvv}. \label{eq:worstcrb}
\end{align}
To derive $\crb   \brc{\{\mC_k\}} \vert \bvv)$, we derive the FIM in \textbf{Theorem-1}.
\begin{theorem}
	\label{th:1}
 Let matrix $\mA_k\brc{\bvv}$ be defined as in \eqref{eq:A_k}, and denote its derivative with respect to $\ell_i$ for $i\in\{x,y,z\}$, as $ \mA_k^i \brc{\bvv} = \frac{\partial}{\partial \ell_i}\mA_k \brc{\bvv}$. The FIM is then specified by a $3\times 3$ matrix whose entry $(i,j)$ is given by 
\begin{align}\label{fisher}
	[\mJ\brc{\bvv}]_{i,j} = 
\sum_{k=1}^K \tr{ \mA_k^i \brc{\bvv} {\mC_k\mC_k^\her} \mA_k^{j\her} \brc{\bvv}}
\end{align}
 where in this indexing $x \equiv 1$, $y \equiv 2$ and $z \equiv 3$. 
\end{theorem}
\begin{proof}
To derive the FIM, note that conditioned on $\bvv$, the matrices $\hat{\Xi}_1, \ldots, \hat{\Xi}_K$ are \ac{iid} Guassian. We can hence write
\begin{subequations}
    \begin{align}
 &p\brc{\hat{\Xi}_1 , \ldots, \hat{\Xi}_K  \vert \bvv } = \prod_{k=1}^K p\brc{\hat{\Xi}_k  \vert \bvv }, \\[-10pt]
		&\propto \prod_{k=1}^K \exp\{ - \norm{\hat{\Xi}_k - \mA_k \brc{\bvv} \brc{\mC_k\mC_k^\her}^{1/2} }_F^2 \}.
	\end{align}
\end{subequations}
	Using the identity $\vert \mathbf{X}\vert_F^2=\tr{\mathbf{X}^\her \mathbf{X}}$, we conclude  equation \eqref{eq-top} on the top of next page.
  \begin{figure*}[t]
	\begin{align}\label{eq-top}
 \log p\brc{\hat{\Xi}_1 , \ldots, \hat{\Xi}_K  \vert \bvv }=
   - \sum_{k=1}^K 
   \tr{\brc{\hat{\Xi}_k - \mA_k \brc{\bvv} \brc{\mC_k\mC_k^\her}^{1/2}}^\her \brc{\hat{\Xi}_k - \mA_k \brc{\bvv} \brc{\mC_k\mC_k^\her}^{1/2}}}.
	\end{align}
  \hrule
  \end{figure*}
	Recall that $\bvv = [\ell_x,\ell_y,\ell_z]^\trp$. Using chain rule, it is readily shown that for $i,j \in \{x,y,z\}$
	\begin{align}
 \frac{\partial^2}{\partial \ell_i \partial \ell_j}\log p&\brc{\hat{\Xi}_1 , \ldots, \hat{\Xi}_K  \vert \bvv }\nonumber
\\[-10pt]& =
- \sum_{k=1}^K 
   \tr{\mA_k^i \brc{\bvv} \brc{\mC_k\mC_k^\her} \mA_k^{j\her} \brc{\bvv}},\nonumber
	\end{align}
 which concludes the proof.
\end{proof}



Computing the \ac{crb} involves inverting the matrix \eqref{fisher} as established in \textbf{Theorem~\ref{th:1}}. However, obtaining an analytic solution for this inversion will not always be feasible. To address this, we seek an alternative approach to establish a lower bound for the worst-case CRB. We begin with \textbf{Lemma~\ref{lem:3}} given below that provides an upper bound on each entry of the FIM.
Using this result, we derive in \textbf{Theorem~\ref{thm:3}} a universal lower bound for the CRB that does not require explicit inversion of the FIM, thus providing an efficient alternative to the expression in \textbf{Theorem~\ref{th:1}}.
\vspace{-0.4cm}
\begin{lemma}\label{lem:3}
    For entry $(i,j)$ of the FIM, we have
    \begin{align*}
		[\mJ\brc{\bvv}]_{i,j} \leq  
  \mathbb{E}\brc{\{\mC_k\}} \sum_{k=1}^K 
   \norm{ \mA_k^i\brc{\bvv}}_F \norm{\mA_k^{j} \brc{\bvv}}_F,
	\end{align*}
 where $\mathbb{E}\brc{\{\mC_k\}} =  {\sum_{k=1}^K 
   \tr{ {\mC_k\mC_k^\her} }} = \sum_{k=1}^K \norm{\mC_k}_F^2$.
\end{lemma}
\vspace{-0.2cm}
\begin{proof}
We start the proof by noting that ${\mC_k\mC_k^\her}$ is positive semi-definite. We hence can write
    \begin{subequations}
        \begin{align}
		\hspace{-.45cm}[\mJ\brc{\bvv}]_{i,j} &= 
  \sum_{k=1}^K 
   \tr{ \mA_k^i \brc{\bvv} {\mC_k\mC_k^\her} \mA_k^{j\her} \brc{\bvv}},\\[-4pt]
   &= \sum_{k=1}^K 
   \tr{  {\mC_k\mC_k^\her} \mA_k^{j\her} \brc{\bvv}\mA_k^i \brc{\bvv}},
   \end{align}
   \begin{align}
   \phantom{\hspace{-.5cm}[\mJ\brc{\bvv}]_{i,j}}
   &\leq \sum_{k=1}^K 
   \tr{  {\mC_k\mC_k^\her} } \norm{\mA_k^{j\her}\brc{\bvv}\mA_k^i \brc{\bvv}}_F,\\[-3pt]
   &\hspace{-.15cm}\leq \sum_{k=1}^K 
   \tr{  {\mC_k\mC_k^\her} } \norm{\mA_k^{j}\brc{\bvv}}_F \norm{\mA_k^i \brc{\bvv}}_F.
	\end{align}
    \end{subequations}
 As all arguments under the summand are positive, we can use Jenesen's inequality to write
     \begin{align}
		\sum_{k=1}^K&\tr{  {\mC_k\mC_k^\her} } \norm{\mA_k^{j}\brc{\bvv}}_F \norm{\mA_k^i \brc{\bvv}}_F\\[-5pt]
   &\leq \brc{\sum_{k=1}^K 
   \tr{  {\mC_k\mC_k^\her} }} \sum_{k=1}^K \norm{\mA_k^{j}\brc{\bvv}}_F \norm{\mA_k^i \brc{\bvv}}_F, \nonumber
	\end{align}
 which concludes the proof.	
\end{proof}
\color{black}

\begin{theorem}[Universal bound on sensing error]\label{thm:3}There exists a constant $\varrho > 0$, such that for any $\{\mC_k \}$, we have 
 \begin{align}
     \crb_{\mathrm{worst}}  \brc{\{\mC_k\}} \geq \crb_L \brc{\{\mC_k\}},
 \end{align}
 where $E\brc{\{\mC_k\}} = \sum_{k=1}^K \norm{\mC_k}_F^2$, $\frac{1}{\bar{\varsigma}^2} = \sum_{k=1}^K \frac{1}{\varsigma_k^2}$,
 and 
 \begin{align*}
     \crb_L \brc{\{\mC_k\}}=\frac{\varrho \bar{\varsigma}^2 }{T E\brc{\{\mC_k\}}}.
 \end{align*}
\end{theorem}
\vspace{-.3cm}
\begin{proof}
See \textbf{Appendix C}.
\end{proof}

This bound is universal in the sense that it does not depend on position $\bvv$, and hence can be used for system design. Interestingly, the bound indicates that the sensing error is reversely proportional to $\norm{\mC_k}_F^2$. This is intuitive, as by increasing $\norm{\mC_k}_F^2$, the \ac{snr} at the devices grows, leading to better estimation of $\bvv$. In the sequel, we adopt this bound as the sensing metric for \ac{moop} in the next section.

\section{MOOP Formulation and Solution }\label{moop2}
Our ultimate goal is to design  $\{\mC_k\}$ and $\mM[t]$, for $k=1,\ldots, K$ and $t=1,\ldots, T$, such that the sensing and aggregation perform efficiently in the considered OTA-FEEL-enabled collaborative ISAC framework. Given the performance metrics derived in the previous section, we formulate the following multi-objective optimization problem:
\begin{subequations}
\begin{align}
	\mathrm{MOOP 1:} &\min_{ \mM[t], \{\mC_k\} } \{ \overline{\mse}\brc{\mM[t], \{\mC_k\} }, \crb_L \brc{\{\mC_k\}} \}\nonumber\\
	&\text{subject to } \; \;  \norm{\mC_k}_{\rm F}^2 \leq P, \; \text{for } k=1,\ldots,\!K,\\
 &\phantom{\text{subject to }} \; \; \; \crb_L \brc{\{\mC_k\}} \leq \varepsilon_0,\\
 &\phantom{\text{subject to }} \; \; \; \overline{\mse}\brc{\mM[t], \{\mC_k\} } \leq \varepsilon_1. 
\end{align}
\end{subequations}
In this problem, two objectives are to be minimized under the transmit power constraints of the devices, sensing, and aggregation quality constraints. \textcolor{black}{
Intuitively, $\mathrm{MOOP 1}$ suppresses destructive impacts of the channel by minimizing both the aggregation and sensing error.} Due to the competing nature of these objectives, it is not possible to optimize both simultaneously: the first objective aims to find a precoder that beamforms towards the \ac{ps} and hence minimizes computation \ac{mse}, while the second objective aims to push the beamforming vectors towards the target, such that the sensing error reduces. The classical approach is to define the so-called Pareto front, which consists of \textit{Pareto-optimal} solutions 
\cite{marler2004survey, miettinen1999nonlinear}.
\vspace{-.3cm}
\begin{definition}
The beamforming matrix design $\{ \mC_k^\star \}$ for $k\in[K]$ is Pareto
optimal if and only if there exists no other design $\mC_k'$
satisfying $\overline{\mse}\brc{\mM[t], \{\mC_k'\} } \leq \overline{\mse}\brc{\mM[t], \{\mC_k^\star\} }$ and $\crb_L \brc{\{\mC_k'\}} <\crb_L \brc{\{\mC_k^\star\}}$, simultaneously.
\end{definition}
Intuitively, a Pareto-optimal solution is a design that is optimal to a \textit{combined} form of those two objectives. The Pareto-front is then the collection of all these optimal solutions. 
Similar to \ac{soop}, Pareto-optimal solutions are either local or global: if the functions and constraints are convex, the global Pareto-optimal solution is feasible. In the non-convex case, however, algorithmic approaches can compute local Pareto-optimal solutions \cite{miettinen1999nonlinear}. 
Various methods exist to determine the Pareto front of an MOOP, among which we can name the Tchebycheff method, scalarization, and $\epsilon$-constraint method. In the sequel, we adopt the $\epsilon$-constraint method, which allows us to convert the MOOP problem into a parameterized SOOP and facilitates the exploration of the Pareto front.

\subsection{MOOP via $\varepsilon$-constraint method}
A classical approach to jointly optimize the two objectives is the $\varepsilon$-constraint method, where one objective is optimized while keeping other objectives within tolerable thresholds. In our setting, we minimize the \ac{mse} objective while constraining the CRB objective to remain within a permissible threshold. This approach aligns with PS's primary role of estimating the aggregated model, with localization as an auxiliary task. Thus, we address the following constrained minimization problem:
\begin{subequations}
\label{eq:KOLLI}
\begin{align}
	&\min_{ \mM[t], \{\mC_k\} } \overline{\mse}\brc{\mM[t], \{\mC_k\} }, \nonumber \\
	&\text{subject to } \; \;  \norm{\mC_k}_{\rm F}^2 \leq P, \; \; \text{for } k=1,\ldots, K,\\
 &\phantom{\text{subject to }} \; \; \; \crb_L \brc{\{\mC_k\}} \leq \varepsilon_0,\\
 &\phantom{\text{subject to }} \; \; \; \overline{\mse}\brc{\mM[t], \{\mC_k\} } \leq \varepsilon_1, \label{eq:CONS_new}
\end{align}
\end{subequations}
where $\varepsilon_0$ and $\varepsilon_1$ that are proportional to the maximum tolerable sensing and aggregation error, respectively. Noting that the $\varepsilon$-constraint \eqref{eq:CONS_new} imposes an upper-bound on the objective function that is being minimized, it is straightforward to conclude that \eqref{eq:KOLLI} is mathematically equivalent to an optimization in which \eqref{eq:CONS_new} is dropped, i.e., the minimal objective is either smaller than $\varepsilon_1$, which makes \eqref{eq:CONS_new} inactive, or larger than $\varepsilon_1$ which makes the feasible set empty. We hence drop the $\varepsilon$-constraint \eqref{eq:CONS_new} in the sequel and assume that $\varepsilon_1$ is always chosen such that the feasible set is non-empty.
Using the definition of $\crb_L \brc{\{\mC_k\}}$, we can rewrite the optimization problem in \eqref{eq:KOLLI} with \eqref{eq:CONS_new} dropped as follows
\begin{subequations}
\begin{align}
	\hspace{-.35cm}\maP_1: &\min_{ \mM[t], \{\mC_k\} } \overline{\mse}\brc{\mM[t], \{\mC_k\} }, \nonumber \\[-5pt]
	&\text{subject to } \; \;  \norm{\mC_k}_{\rm F}^2 \leq P, \; \; \text{for } k=1,\ldots, K,\\[-4pt]
 &\phantom{\text{subject to }} \; \; \; 
    \sum_{k=1}^K \norm{\mC_k}_{\rm F}^2 \geq \varepsilon^{-1},
\end{align}
\end{subequations}
where $\varepsilon=\frac{\varepsilon_0 T}{\varrho \bar{\varsigma}^2}$. The problem $\maP_1$ is non-convex in terms of $\mM[t]$ and $\{\mC_k\}$.
\color{black}
To address this, we employ a \ac{bcd} approach, decomposing $\maP_1$ into two marginal subproblems: $\mathcal{M}_1$, which optimizes $\mM[t]$ for fixed $\mC_k$, and $\mathcal{M}_2$, which optimizes $\mC_k$ for fixed $\mM[t]$. These subproblems are solved iteratively, where the solution of each subproblem updates the variables for the other until convergence, leading to a balanced solution for $\maP_1$. The corresponding marginal problems are formulated as:\color{black}
\begin{subequations}
\begin{align}
\mathcal{M}_1: &\min_{ \mM[t] } \overline{\mse}\brc{\mM[t], \{\mC_k\} }, \nonumber
\end{align}
and
\begin{align}
\hspace{-.15cm}\mathcal{M}_2: &\min_{ \{\mC_k\} } \overline{\mse}\brc{\mM[t], \{\mC_k\} }, \nonumber\\
	&\text{subject to } \; \;  \norm{\mC_k}_{\rm F}^2 \leq P \; \; \text{for } k=1,\ldots, K,\\
 &\phantom{\text{subject to }} \; \; \; 
    \sum_{k=1}^K \norm{\mC_k}_{\rm F}^2 \geq \varepsilon^{-1}.
\end{align}
\end{subequations}
\color{black}The first marginal problem, $\mathcal{M}_1$, minimizes the MSE with respect to the PS’s receive beamforming matrix $\mM[t]$, treating $\mC_k$ as fixed. The solution to this subproblem is determined analytically in the next section. The second marginal problem, $\mathcal{M}_2$, optimizes the device transmit beamformers $\mC_k$ to minimize MSE while ensuring power constraints and a lower bound on sensing accuracy. This marginal problem is non-convex, making it challenging to solve directly. To overcome this, we reformulate it into a variational form and solve it using iterative projection, as detailed in Section IV.C.\color{black}
\begin{algorithm}[t]
\caption{MOOP1 Solution} \label{alg:1}
\begin{algorithmic}[1]
    \Require $\mH_k$, $\mW_k$, and $p_k[t]$ for $k\in [K]$, power $P$, noise variance $\sigma^2$, and threshold $\varepsilon$
    \While{not converged}\hspace{-.5cm}\Comment{Using some convergence criteria}
    \For{$t=1:T$}
    \State Compute $\mM^\star [t]$ via $\eqref{eq:M_star}$
    \EndFor
    \State Compute $\alpha_k$ and $\mC_k$ for $k \in [K]$ by solving $\mathcal{M}_\star$\Comment{Convex programming}
    \For{$k=1:K$}
    \State Set $\mC^\star_k \leftarrow \sqrt{\alpha_k}\mC_k/\norm{\mC_k}_F$
    \EndFor
    \EndWhile
    \Ensure $\mM^\star [t]$ and $\mC^\star_k$
\end{algorithmic}
\end{algorithm}

\subsection{First Marginal Problem}
The marginal problem $\mathcal{M}_1$ is a standard quadratic problem. Considering $\overline{\mse}\brc{\mM[t], \{\mC_k\} }$, $\mathcal{M}_1$ is rewritten as follows:
\begin{align}
    \min_{ \mM[t] } \frac{1}{T} \sum_{t=1}^T \sum_{k=1}^{K}  
 &\norm{ {\mM^\her [t] \mH_k \mC_k   p_k[t]\!-\!\mW_k } }_F^2\nonumber\!+\!\frac{\sigma^2}{T} \sum_{t=1}^T\!\norm{\mM [t]}_{\rm F}^2 ,
\end{align}
which reduces to a sequence of $T$ decoupled problems for $\mM[1], \ldots, \mM[T]$. For $t\in [T]$, this problem reads
\begin{align}
    &\min_{ \mM[t] }  \sum_{k=1}^{K}  
 \norm{ {\mM^\her [t] \mH_k \mC_k   p_k[t] - \mW_k } }_F^2
 + \sigma^2 \norm{\mM [t]}_{\rm F}^2,\nonumber
\end{align}
whose solution is given by:
\begin{align}\label{eq:M_star}
\mM^\star [t]\!=\!\brc{\!\sum_{k=1}^{K}  
  \mH_k \mC_k \mC^\her_k \mH^\her_k\!+\! \sigma^2 \mI_N\!}^{\!-1}\!\!\!\!\sum_{k=1}^K p_k[t] \mH_k \mC_k \mW_k^\her.
\end{align}
\vspace{-0.5cm}
\subsection{Second Marginal Problem}
The second marginal problem minimizes the objective with respect to $\{ \mC_k\}$. We can rewrite $\mathcal{M}_2$ as follows:
\begin{subequations}
\begin{align}
&\min_{ \{\mC_k\} } \frac{1}{T} \sum_{t=1}^T \sum_{k=1}^{K} 
 \norm{ {\mM^\her [t] \mH_k \mC_k   p_k[t] - \mW_k } }_F^2, \nonumber\\
	&\text{subject to } \; \;  \norm{\mC_k}_{\rm F}^2 \leq P, \; \; \text{for } k=1,\ldots, K,\\[-3pt]
 &\phantom{\text{subject to }} \; \; \; 
    \sum_{k=1}^K \norm{\mC_k}_F^2 \geq \varepsilon^{-1}.
\end{align}
\end{subequations}
Unlike $\mathcal{M}_1$, this problem is not convex as the objective represents a difference-of-convex functions which is challenging to solve \cite{boyd2004convex}. To tackle this problem, we thus perform the following steps:
1) We transform the problem into variational form using an auxiliary variable. 2) We relax the variational problem and use iterative projection to solve the problem.
These two steps are explained below.

\textit{1) Variational Form:} Let us define the auxiliary variable $\alpha_k = \norm{\mC_k}_F^2$. Hence, the marginal problem is rewritten as:
\begin{subequations}
\begin{align}
&\min_{ \{\alpha_k,\mC_k\} } \frac{1}{T} \sum_{t=1}^T \sum_{k=1}^{K} 
 \norm{ {\mM^\her [t] \mH_k \mC_k   p_k[t] - \mW_k } }_F^2, \nonumber\\[-3pt]
	&\text{subject to } \; \;  \alpha_k - P \leq 0 \; \; \text{for } k=1,\ldots, K,\\[-3pt]
 &\phantom{\text{subject to }} \; \; \; 
    \varepsilon^{-1} - \sum_{k=1}^K \alpha_k \leq 0, \\[-2pt]
    &\phantom{\text{subject to }} \; \; \; 
    \norm{\mC_k}_F^2 - \alpha_k = 0 \; \; \text{for } k=1,\ldots, K.
\end{align}
\end{subequations}
In this variational form, the first two constraints are linear that can be handled by linear programming. However, the last constraint imposes a \textit{unit-modulus} restriction on the problem.

\textit{2) Iterative Projection:} To approximate the solution efficiently, we relax the problem into
\begin{subequations}
\begin{align}
\mathcal{M}_\star: &\min_{ \{\alpha_k, \mC_k\} } \frac{1}{T} \sum_{t=1}^T \sum_{k=1}^{K} 
 \norm{ {\mM^\her [t] \mH_k \mC_k   p_k[t] - \mW_k } }_F^2,\nonumber \\[-3pt]
	&\hspace{-.2cm}\text{subject to } \; \;  \alpha_k - P \leq 0, \; \; \text{for } k=1,\ldots, K,\\[-3pt]
 &\hspace{-.4cm}\phantom{\text{subject to }} \; \; \; 
    \varepsilon^{-1} - \sum_{k=1}^K \alpha_k \leq 0, \\[-3pt]
    &\hspace{-.4cm}\phantom{\text{subject to }} \; \; \; 
    \norm{\mC_k}_F^2 - \alpha_k \leq 0, \; \; \text{for } k=1,\ldots, K.
\end{align}
\end{subequations}
This problem is convex and can hence be solved via a convex programming algorithm. Nevertheless, the equality constraint of the variational problem does not necessary hold. To enforce this constraint to the solution, we invoke iterative projection algorithm \cite{bertsekas1997nonlinear} in each iteration of the \ac{bcd} loop, i.e.,
\begin{itemize}
    \item We solve the relaxed problem for $\alpha_k,\mC_k$.
    \item We normalize $\mC_k$ as $\mC_k \leftarrow \sqrt{\alpha_k}\mC_k/\norm{\mC_k}_F$ to satisfy the equality constraint.
\end{itemize}
We keep iterating till a convergence criterion is satisfied. 
\textbf{Algorithm~\ref{alg:1}} outlines the solution approach for MOOP1, while \textbf{Algorithm~\ref{alg:2}} demonstrates the CollabSenseFed Algorithm as an OTA-FEEL-enabled collaborative ISAC framework.

\subsection{Discussions on Overhead, Complexity and Convergence}
\color{black}
To quantitatively evaluate the advantages of distributed sensing over centralized sensing , we introduced the concept of sensing signaling load (SSL), which measures the communication overhead involved in each approach. For centralized sensing, as outlined in \cite{behdad2022power}, the SSL corresponds to forwarding raw sensing signals from all devices to the \ac{ps}, mathematically expressed as  $\mathcal{T}_{\mathrm{CS}} = 2 K M s$, where $s$ represents the number of time-domain samples per antenna collected by the PS for signal reconstruction. According to the Nyquist sampling theorem, $s\approx 2 B T_s$, where $B$ is the sensing bandwidth and $T_s$ is the sensing duration in seconds.
On the other hand, the SSL of the distributed approach is $\mathcal{T}_{\mathrm{DS}} = d R/\tau$, where $d$ is the dimension of the target parameter, $R$ is the number of iterations required by the distributed algorithm, and $\tau$ is the number of local iterations at the devices. In practice, $dR/\tau \ll s$   is typically a small integer, e.g., $d=3$ in localization, and  $\tau$ could be set relatively large for efficient distributed  approaches.  In OTA computation, the devices transmit over the same frequency band and time to perform the computation  in the analog domain. Thus, the SSL does not scale with $K$.
\color{black}

The complexity of MOOP solution per iteration is determined by three main steps. First, computing $\mM^\star[t]$ involves summing $K$ terms, where each term requires matrix multiplications (e.g., $\mH_k \mC_k \mC_k^\her \mH_k^\her$) and a matrix inversion. The cost for these operations is $\mathcal{O}(KNM^2)$ for the multiplications, $\mathcal{O}(KNMI)$ for additional products involving $\mW_k$, and $\mathcal{O}(N^3)$ for the inversion, leading to a total complexity of $\mathcal{O}(KNM^2 + KNMI + N^3)$. Second, solving the convex problem $\mathcal{M}_\star$ requires minimizing a Frobenius norm-based objective function with $K$ variables, each involving matrices of size $M \times I$. This results in a complexity of $\mathcal{O}((K(MI + 1))^3)$ due to the cubic scaling of convex optimization algorithms. Third, normalizing $\mC_k$ involves computing its Frobenius norm and scaling, costing $\mathcal{O}(KMI)$. 
Combining these, the per-iteration complexity is approximately $\mathcal{O}(KNM^2 + N^3 + (KMI)^3)$. Assuming $L$ iterations for convergence, the overall complexity becomes approximately 
$\mathcal{O}(L(KNM^2 + N^3 + (KMI)^3))$. \textcolor{black}{Throughout the numerical experiments, we evaluated the average runtime per epoch of \textbf{Algorithm~\ref{alg:2}} against various system parameters and fitted the results using curve-fitting algorithms to validate the complexity order, e.g., we have observed cubic scaling with the number of devices and antennas.} 

It is straightforward to show that CollabSenseFed converges to local minimum of the training losses and  the sensing objective. This follows from the fact that \textbf{Algorithm~\ref{alg:1}} guarantees the aggregation error to be bounded \footnote{\color{black}Algorithm~\ref{alg:2} assumes full device participation, but in practice, dropouts can bias aggregation based on data distribution, affecting both learning and sensing. This can be mitigated using fault-tolerant FL strategies like weighted aggregation and dropout-aware updates \cite{qian2024dropfl, sun2023mimic}.\color{black}}.
Hence, one can use \cite[Theorem~2.2]{friedlander2012hybrid} to show that with sufficiently small learning rate choices, the CollabSenseFed converges to a local minimum of the empirical loss. The proof follows the same approach in \cite{bereyhi2023device,asaad2024joint} and the references therein.
Due to lack of space, we skip the details and refer interested readers to \cite{asaad2024joint,bereyhi2023device}.

\begin{algorithm}[t!]
\caption{CollabSenseFed Algorithm}\label{alg:2}
    \begin{algorithmic}[1]
        \Require Initiate some random location $\bvv$, set $d$ to the learning model size and $t=0$.
        \State Compute $\mC_k$ and $\mM[t]$ via \textbf{Algorithm~\ref{alg:1}}.
        \For{multiple communication rounds}
        \For{$t=1, \ldots, d$ intervals}
        \For{Device $k = 1:K$}
        \State Set $\bss_k [t] = \sum_{i=1}^I \bc_k^{i} g_k^i [t]$;
        \State Transmit $\bxx_k[t] = \bss_k[t] p_k[t]$ to the PS.
        \State 
        Compute $\hat{\Xi}_k$ from reflections with \eqref{eq:suff_stat}.
          \State Set~$g_k^1[t+1]$ to~be an entry~of $\nabla \ell_k (\bvv)$. 
          \For{$i=2,\ldots,I$}
        \State Compute local gradient $ g_k^i [t+1]$
        \EndFor
        \EndFor
        \State PS receives $\byy[t]$ as in \eqref{eq:yPS}.
        \State PS computes $\br [t] =  \mM^\her [t]  \byy[t]$ and send it back.
        \State Update the corresponding entry of $\bvv$ with $r_1(t)$ 
        \For{$i=2,\ldots,I$}
        \State Set entry $t$ of global gradient of model $i$ to $r_i [t]$
        \EndFor
        \State Update: $t = t+1$;
        \EndFor
        \For{$i=2,\ldots,I$}
        \State Update model $i$ using its global gradient
        \EndFor 
        \EndFor
    \end{algorithmic}
\end{algorithm}
\color{blue}
\begin{figure*}[t]
  \centering
  \begin{minipage}{0.32\textwidth}
    \centering
    \includegraphics[scale=.89]{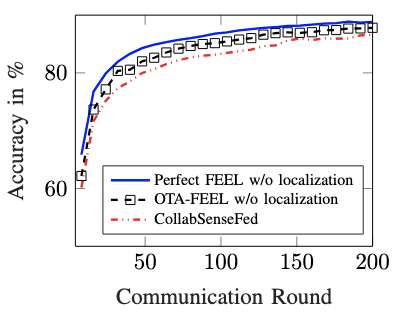}
    \caption{Effect of sensing on MNIST accuracy}
    \label{fig:63}
  \end{minipage}
  \hfill
  \begin{minipage}{0.32\textwidth}
    \centering
   \includegraphics[scale=.88]{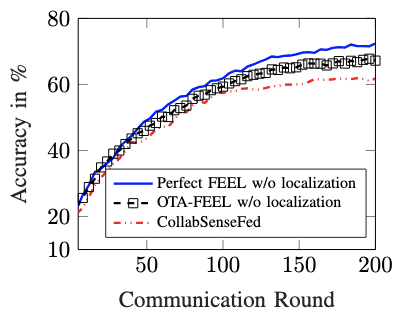}
	\caption{Effect of sensing on CIFAR accuracy}
  \label{fig:8}
  \end{minipage}
  \hfill
  \begin{minipage}{0.32\textwidth}
    \centering
    \includegraphics[scale=.89]{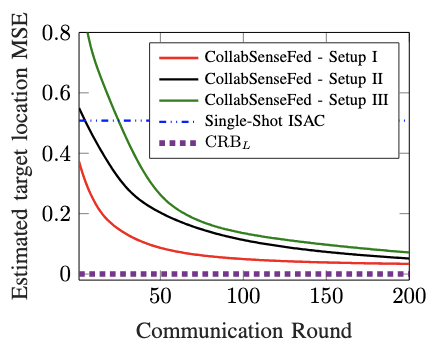}
	\caption{Target location MSE vs. iterations}
   \label{fig:9}
  \end{minipage}
  \vspace{-0.4cm}
\end{figure*}

\color{black}
\section{Numerical Results and Discussions}

\subsection{Experimental Set-up and Benchmarks}
We consider $K$ devices and a \ac{ps}, where the devices collaboratively train a global model and localize a target. In all simulations, devices are randomly placed within a ring centered at the \ac{ps}, with inner radius $R_{\rm in} = 50$~m and outer radius $R_{\rm out} = 100$~m. Their locations are distributed, such that they all lie within an arch whose angle is smaller than $20^\circ$. We assume $K=15$ devices, with $M=4$ antennas at each device and $N=16$ antennas at the \ac{ps}, unless otherwise stated. A target is located randomly within a ring with inner radius $R_{\rm in} = 100$~m and outer radius $R_{\rm out} = 110$~m around the \ac{ps}, in the same arch that devices are located, at an altitude that is uniformly chosen from interval $[0,3]$~m. The array response is specified by \eqref{array_resp} with $\alpha_0$ being set to $1$. The aggregation of the parameters of the local models being trained at the devices,  as well as the sensing gradients, i.e., the gradients being exchanged to solve the joint \ac{ml} localization problem, are carried out over-the-air. To mutually orthogonalize the pulses of the devices, we use DFT pulses of length $T \geq K$. 

The uplink channels are considered to experience \ac{iid} Rayleigh fading with a coherence time interval that includes four communication rounds, i.e., after four rounds of aggregating global model for the learning task. For each realization of this setting, we solve the underlying design problem via \textbf{Algorithm~\ref{alg:1}}. 
We use the samples collected by the devices from target reflections in each frame of $T$ transmissions to estimate the location of the target. To this end, the devices use  \textbf{Algorithm~\ref{alg:2} }to solve the joint \ac{ml} estimation problem in a distributed fashion.
The details of the learning and sensing settings are represented in the sequel.  \textcolor{black}{We also consider the \textit{single-shot} localization approach from \cite{liu2022toward}, where devices compute local sufficient statistics in the first round and send them to the PS for averaging to estimate the target location. After this initial step, the sensing task ends, and only learning updates are transmitted in subsequent rounds. In contrast, the proposed CollabSenseFed scheme performs sensing iteratively: devices continuously extract sensing-related gradients and transmit them alongside learning updates over multiple rounds. To ensure a fair comparison, both schemes are evaluated under identical resource constraints, with the same number of uplink transmissions (i.e., channel uses) and the same total transmit power. In CollabSenseFed, this power is split between sensing and learning throughout the process, whereas in the single-shot baseline, it is fully devoted to learning after the first round}



\textcolor{black}{The proposed joint sensing-learning approach significantly improves sensing accuracy, with only a modest degradation in learning performance due to interference from sensing gradient transmissions. In contrast, the single-shot baseline achieves slightly better learning accuracy by allocating all resources to gradient exchange after the first round, but this comes at the cost of substantially reduced sensing performance. These trade-offs are further demonstrated and quantified in the following through numerical results.}

\subsection{Learning Setting}
We evaluate two learning setups on MNIST and CIFAR-10. For MNIST, a two-layer MLP is trained on 60,000 images of size  \(28 \times 28\), with 10,000 used for testing \cite{xiao2017fashion}. The training and testing data sets are randomly shuffled and unevenly distributed across devices. To simulate non-uniform data distribution in federated learning, we assign random percentages to each client using the Dirichlet distribution \textcolor{black}{with concentration parameter $\alpha=0.4$ as in \cite{marfoq2021federated}, which introduces a moderate level of data heterogeneity, ensuring that different clients receive imbalanced class distributions. Note that non-uniform distribution mainly affects the learning task by increasing aggregation variance, which can slow convergence. However, sensing performance remains largely independent of data heterogeneity, as it is driven by physical signal reflections rather than class distributions.} 
For CIFAR-10 classification \cite{krizhevsky2009learning}, we train a deep CNN on $50,000$ RGB images of size \(32 \times 32\), with 10,000 used for testing. The network comprises three convolutional blocks—each with two \(3 \times 3\) convolutional layers (using $32$, $64$, $128$, and $256$ filters), ReLU activation, and \(2 \times 2\) max-pooling—followed by fully connected layers of sizes 1024 and 512, and a final 10-output softmax. In each communication round, the devices iterate five local updates on their local models via \ac{sgd}. The model parameters are then transmitted jointly with the gradients of localization loss using the uplink beamforming matrices $\mC_k$. We consider $I=2$, where $i=1$ and $i=2$ refer to the localization and learning tasks, respectively. 


\subsection{Results and Discussions}

\par \textit{1) Performance of CollabSenseFed Algorithm}: Fig.~\ref{fig:63} shows the performance of MNIST classification and compares the results with two ideal cases, namely, \textit{perfect FEEL without localization} and \textit{OTA-FEEL without localization}. These cases consider the learning task being performed individually (no target for sensing) with the former considering perfect orthogonal and noiseless uplink communication, and the latter considers OTA-FEEL. As the figure shows, CollabSenseFed tracks both cases closely, indicating a practically-negligible impact of target sensing on learning quality \cite{yang2020federated,bereyhi2023device}. 
The small gap between CollabSenseFed and learning-only cases can be attributed to inter-task interference. In fact, in this case, the interference by the localization gradients on the aggregated parameters leads to higher aggregation error and hence degrades performance of the trained classifier. 
This degradation is in fact the price paid to integrate a distributed sensing functionality into our system.
Fig.~\ref{fig:8} shows the results for training of the CNN architecture over the CIFAR-10 dataset. The result shows a similar trend with a larger gap between the baselines, that comes from the fact that the classification task in CIFAR-10 is more complex.
Fig.~\ref{fig:9} investigates the sensing performance of the CollabSenseFed and its variants against  single-shot distributed ISAC approach. The three variants are given as: \\$\bullet$~\textit{CollabSenseFed - Setup I}: The devices only estimate the target location via CollabSenseFed while communicating their sensing gradients over perfect (noiseless and orthogonal) uplink channels, i.e., no learning task is involved. 
\\$\bullet$~\textit{CollabSenseFed - Setup II}: Only sensing is considered (similar to Setup I), however, the aggregation of the gradients is performed directly over-the-air. 
\\$\bullet$~ \textit{CollabSenseFed - Setup III}: CollabSenseFed is used to perform the joint sensing and learning over-the-air. 
\begin{figure*}[t]
  \centering
  \begin{minipage}{0.32\textwidth}
    \centering
\includegraphics[scale=.4]{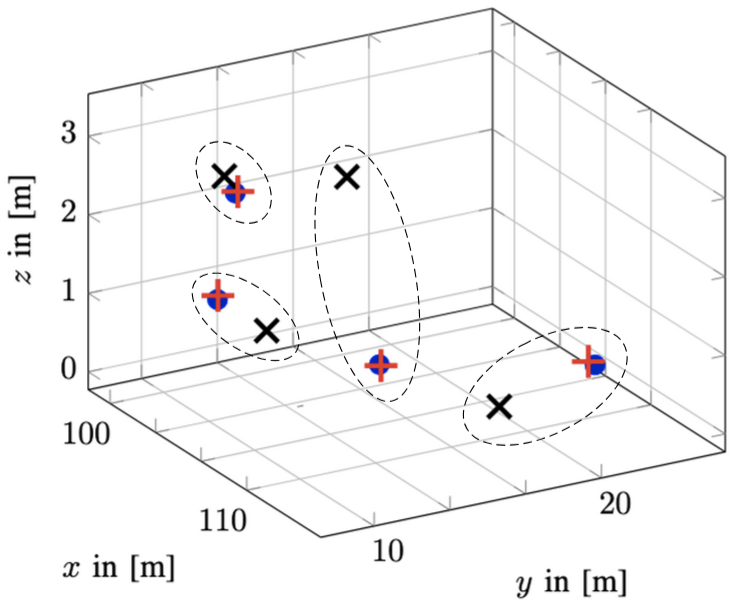}
\caption{Target location estimation performance.}
\begin{tabular}{r@{: }l}
  {\color{red}\textbf{\(\boldsymbol{+}\)}} & \small Estimated target location \\
  {\color{black} \scalebox{1.3}{\(\star\)}} & \small Ground-truth target location \\
  {\(\boldsymbol{\times}\)} & \small Single-Shot Sensing 
\end{tabular}

\label{fig:FigEnd}
  \end{minipage}
  \hfill
  \begin{minipage}{0.32\textwidth}
    \centering
\includegraphics[scale=.85]{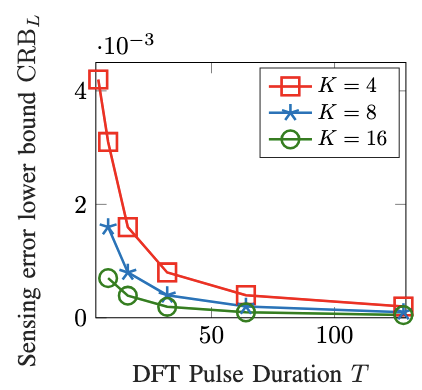}
	\caption{Lower bound on variance of sensing error.}
 \label{fig:crbvstime}
  \end{minipage}
  \hfill
  \begin{minipage}{0.32\textwidth}
    \centering
    \includegraphics[scale=.87]{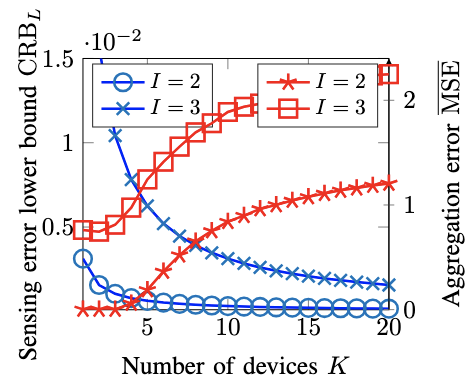}
	\caption{Sensing error and aggregation error vs. $K$.}
  \label{fig:MSE-CRBvsUser}
  \end{minipage}
  \vspace{-0.5cm}
\end{figure*}
\\
Fig.~\ref{fig:9} shows the \ac{mse} of~the estimated target, denoted as $\mathrm{MSE} = \Ex{\norm{\mathbf{v} - \hat{\mathbf{v}}}^2 }{ }$, plotted against communication rounds. Here, we consider CIFAR-10 classification as the learning task.
As the results show, CollabSenseFed gradually reduces the sensing error towards the $\mathrm{CRB}_L$ in all three cases. Among the three setups, Setup III converges the slowest, which is due to the interference by the learning task. The three setups, however, converge to the joint Cram\'er-Rao bound after enough number of iterations. This is intuitive as CollabSenseFed solves the \ac{ml} estimation problem in a distributed fashion and hence at its converging point it performs close to the joint Cram\'er-Rao bound. \textcolor{black}{
The figure further demonstrates the trade-off between sensing and communication: with a larger number of communication rounds, we can achieve more accurate estimation of the location, which is an intuitive observation.}
The figure shows a drastic enhancement achieved by CollabSenseFed as compared to benchmark \textit{single-shot} localization. This follows from the fact that in single-shot sensing, the devices estimate the target location from an insufficient amount of data. This leads to large localization error which does not vanish by averaging. On the other hand, CollabSenseFed communicates the local gradients together with its uplink transmissions and updates gradually its estimation through time via moving towards the joint \ac{ml} estimation. The CollabSenseFed shows its best performance with perfect communication links, which is further consistent with intuition. Additionally, we examine the scenario where users transmit their local gradients over-the-air, focusing solely on the localization task, i.e., $I=1$. In this case, the performance surpasses that of CollabSenseFed. This improvement is attributed to the absence of inter-task interference, which simplifies the learning process. 
Fig.~\eqref{fig:FigEnd} shows the target location estimates by CollabSenseFed, which closely match the ground truth, demonstrating its superior performance. In contrast, the single-shot localization method shows large deviations from the true location.
\par \textit{2) Impact of DFT Pulse Duration:}
Fig.~\ref{fig:crbvstime} shows the sensing error achieved by Algorithm~\ref{alg:2} against the pulse duration $T$ for various choices of $K$. As the figure shows, the sensing error drops monotonically against pulse duration. This is intuitive as by increasing the pulse duration for a fixed number of devices $K$, the variance of the sufficient statistics collected by the devices reduces, and therefore the estimation error decreases. Furthermore, the figure indicates that the sensing error also decreases as the number of devices $K$ increases. This is also intuitive, as a greater number of devices results in more samples, and thereby enhanced localization accuracy. 
\par\textit{3) Impact of the Number of Users:} We next investigate the variation of the aggregation error and sensing error bound against $K$ for various choices of $I$. 
While simulations focus on I=2I=2—one task for over-the-air classification training and the other for distributed localization—this figure treats II as a system parameter to study its impact on Algorithm~\ref{alg:1}’s solution. As Fig.~\ref{fig:MSE-CRBvsUser} demonstrates, the sensing error bound reduces against $K$ which is consistent with our observation in Fig.~\ref{fig:crbvstime}. This figure also demonstrates that by increasing the number of learning tasks, the achieved CRB increases. This behavior follows the mutual dependency of the aggregation and sensing error in the multi-objective optimization, where by larger number of tasks the interference increases and hence the sensing error increases.  Also, we can see that the aggregation error increases against $K$. With larger number of devices, the interference increases, and hence the aggregation error grows large. We further observe that the aggregation error increases with the number of tasks which confirms our initial illustrations. In fact this follows from higher inter-task interference, and can implicitly lead to higher sensing error.
\begin{figure*}[t]
  \centering
  \begin{minipage}{0.32\textwidth}
    \centering
    \includegraphics[scale=.85]{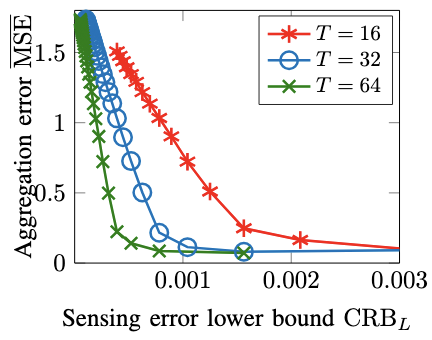}
    \caption{Pareto front: trade-off between sensing and aggregation error for different $T$.}
    \label{fig:pareto}
  \end{minipage}
  \hfill
  \begin{minipage}{0.32\textwidth}
    \centering    
    \includegraphics[scale=.83]{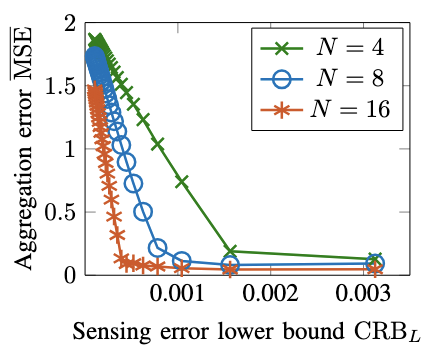}
    \caption{Pareto front: trade-off between sensing and aggregation error for different $N$.}
    \label{fig:paretovsN}
  \end{minipage}
  \hfill
  \begin{minipage}{0.32\textwidth}
    \centering
     \includegraphics[scale=.85]{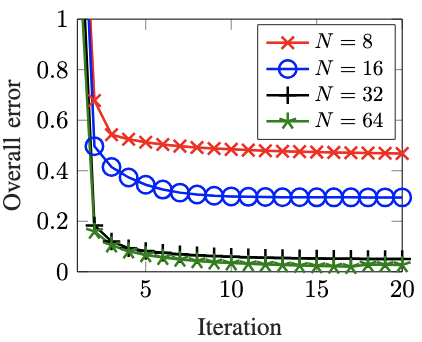}
    \caption{Convergence behavior of CollabSenseFed algorithm.}
    \label{fig:convergence}
  \end{minipage}
    \vspace{-0.5cm}
\end{figure*}
\par \textit{4) Impact of DFT Pulse Duration and Transmit Antennas on Pareto-Front:}
Figs.~\ref{fig:pareto} and \ref{fig:paretovsN} show the Pareto front of the MOOP, illustrating the trade-off between aggregation and sensing errors versus the number of PS antennas and DFT pulse duration, respectively. The shape of the Pareto-front implies from the fact that both the objectives are set to be minimized. In both figures, it can be observed that as the sensing error decreases, the aggregation error tends to increase,  highlighting the inherent trade-off between these two metrics.
This inverse relationship indicates that improving sensing accuracy comes at the cost of higher aggregation error. Additionally, by increasing the number of transmit antennas or expanding the pulse duration, the aggregation error further reduces and sensing error (marginally) improves. 
\par \textit{5) Convergence Analysis:}
Fig. \ref{fig:convergence} shows the convergence behavior of the proposed algorithm as a function of the number of PS antennas. In this figure, the overall error, defined as the average of the sensing error and aggregation error, is plotted against the number of iterations for different antenna configurations. As observed,  the proposed algorithm converges very fast, i.e., after less than 10 iterations. Additionally, the overall error shows a rapid drop with increasing $N$, which can be attributed to the greater degrees of freedom from larger receive arrays that enhance beamforming for model aggregation at the PS. Increasing the number of \ac{ps} antennas improves interference cancellation, thus reducing overall error.
\par \color{black}\textit{6) Centralized vs Distributed  Sensing Overhead:} 
   \begin{figure}
    \centering
    \includegraphics[scale=.85]{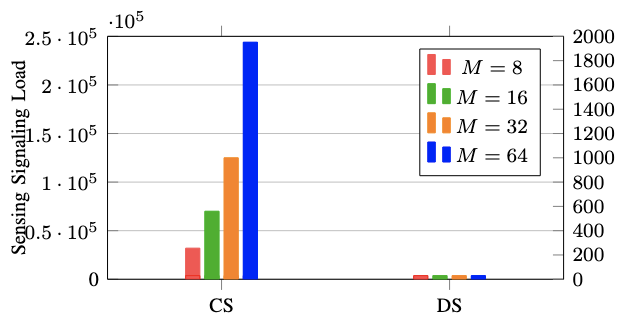}
    \caption{Sensing signaling load vs. number of antenna elements at each device.}
        \label{fig:centralized}
        \vspace{-.4cm}
 \end{figure}
The bar plot in Fig.~\ref{fig:centralized} illustrates the SSL for CS and our proposed DS approach with over-the-air computation, considering varying numbers of antennas $M = \{8, 16, 32, 64\}$ and a sample size $s = 200$. The left vertical axis, ranging from 0 to 250,000 real scalars, displays the CS overhead, which increases linearly with $M$ (e.g., 32,000 for $M = 8$ to 244,000 for $M=64$). The right vertical axis, ranging from 0 to 2000 with ticks every 200, shows the DS overhead, constant at 30 across all $M$ values due to its independence from the number of antennas. The figure 
highlights the significant overhead reduction achieved by distributed sensing in this configuration. 
\color{black}
\vspace{-.6cm}
\section{Conclusion}
This work demonstrated the potential of integrating localization into the \ac{ota-fl} framework by proposing  a novel approach for joint target localization and distributed learning. The proposed scheme builds on the FEEL framework to jointly localize a target using samples collected across the network. Numerical results show that, unlike conventional methods limited by individual estimators, CollabSenseFed progressively reduces sensing error toward the Cram\'er-Rao bound. These findings suggest several possible directions for future research, including natural extensions to distributed mobility sensing and multi-target scenarios. \textcolor{black}{Investigating the impacts of other classical techniques, such as user dropout, is another interesting direction for future work.}  \textcolor{black}{ Analyzing the robustness of the proposed scheme under practical impairments, such as carrier frequency offset (CFO), phase noise, or hardware imperfections, represents an important  direction for future.}

\appendices
\section{ } 
 We start with the first limit: since $\tr{\tilde{\mR}_k [\tau]}$ decays against $\tau$, as $T\rightarrow\infty$ the sequence $\bxx_k [t] \bxx_k^\her [t]$ contains infinitely large number of uncorrelated realizations. Therefore, using the weak law of large numbers, we can write
 \begin{subequations}
     \begin{align}
		\mR_{k,k} &= \frac{1}{P T} \sum_{t=1}^T  \bxx_k [t] \bxx_k^\her [t] \rightarrow \Ex{ \frac{1}{T} \sum_{t=1}^T \bxx_k [t] \bxx_k^\her [t] },\nonumber\\[-4pt]
		&= \frac{1}{P}\Ex{ \frac{1}{T} \sum_{t=1}^T \bss_k [t] \bss_k^\her [t]  \abs{p_k[t]}^2},\nonumber\\[-4pt]
		&= \frac{1}{P} \left(\frac{1}{T} \sum_{t=1}^T  \Ex{\bss_k [t] \bss_k^\her [t] } \right)
		= \frac{1}{P} \tilde{\mR}_k[0].\nonumber
	\end{align}
 \end{subequations}
	To prove the validity of \eqref{eq:eq2}, let us expand the sample covariance matrix $\Ex{\mR_{\ell,k}}  = \frac{1}{PT}  \sum_{t=1}^T \Ex{ \bxx_\ell [t] \bxx_k^\her [t]}\stackrel{\mathrm{(a)}}{=} \frac{1}{PT}  \sum_{t=1}^T \Ex{ \bss_\ell [t]} \Ex{\bss_k^\her [t]} p_\ell [t] p_k^* [t]
		=0$
where (a) follows from the independence of the sequences $\bss_k[t]$ and $\bss_\ell[t]$. Define $f_{\ell k}[t]=p_\ell [t] p_k^* [t] $. 
We next compute the sum-variance of the entries of sample covariance matrix as given in \eqref{eq:DD} at the top of the next page. 
\begin{figure*}
	 \begin{subequations}\label{eq:DD}
 \begin{align}
 \vspace{-1cm}\Ex{\norm{\mR_{\ell,k}}_F^2}&\!=\! \Ex{\tr{\mR_{\ell,k}\mR_{\ell,k}^\her}}\!=\!\frac{1}{P^2T^2} \Ex{  \tr{ \brc{\sum_{t=1}^T \bxx_\ell [t] \bxx_k^\her [t]}
   \brc{\sum_{t=1}^T \bxx_k [t] \bxx_\ell^\her [t]}}},\\[-7pt]
		&\!=\! \frac{1}{P^2T^2} \sum_{t_1=1}^T \sum_{t_2=1}^T \tr{ \Ex{\bss_\ell [t_1] \bss_k^\her [t_1] \bss_k [t_2] \bss_\ell^\her [t_2] }} f_{\ell k}[t_1] f^*_{\ell k}[t_2],\\[-3pt]
		&\!\stackrel{\mathrm{(a)}}{=}\!\frac{1}{P^2 T^2} \sum_{t_1=1}^T \sum_{t_2=1}^T  \Ex{\bss_k^\her [t_1] \bss_k [t_2]  } \Ex{ \bss_\ell^\her [t_1] \bss_\ell [t_2] } f_{\ell k}[t_1] f^*_{\ell k}[t_2],\\[-3pt]
		&\!=\! \frac{1}{P^2 T^2} \sum_{t_1=1}^T \sum_{t_2=1}^T \tr{\tilde{\mR}_k[t_1-t_2]} \tr{\tilde{\mR}_{\ell} [t_2-t_1]}   f_{\ell k}[t_1] f^*_{\ell k}[t_2],\\[-4pt]
		&\!\stackrel{\mathrm{(b)}}{\leq} \!\frac{1}{P^2 T^2} \sum_{t_1=1}^T \sum_{t_2=1}^T \tr{\tilde{\mR}_k[0]} \tr{\tilde{\mR}_{\ell} [0]}  f_{\ell k}[t_1] f^*_{\ell k}[t_2]\!=\!\frac{1}{P^2 T^2} \tr{\tilde{\mR}_k[0]} \tr{\tilde{\mR}_{\ell} [0]}  \brc{\sum_{t=1}^T f_{\ell k}[t] }^2\!\!=\! 0.
	\end{align}
  \end{subequations}
  \hrule
  \end{figure*}
In \eqref{eq:DD}, the identity (a) follows the independency of $\bss_k[t]$ and $\bss_\ell[t]$, and the inequality (b) is the result of $\tr{\tilde{\mR}_k [\tau]}$ decaying. This indicates that $0 \leq \Ex{\norm{\mR_{\ell,k}}_F^2} \leq 0$ implying $\Ex{\norm{\mR_{\ell,k}}_F^2} = 0$.
By the extended central limit theorem, $\mR_{\ell,k}$ converges to a zero-mean, zero-variance Gaussian random variable as $T \to \infty$.
 \section{ } 
 The proof follows directly from the definition $\tilde{\mR}_k[\tau] = \Ex{\bss_k [t] \bss_k^\her [t-\tau]}$. Plugging $\bss_k[t]$ into the definition, we have
\begin{subequations}
\begin{align}
	\vspace{-1.5cm}\tilde{\mR}_k[\tau] 
	&= \Ex{ \mC_k \bgg_k [t] \bgg_k^\her [t-\tau] \mC_k^\her },\\
	&= \mC_k \Ex{\bgg_k [t] \bgg_k^\her [t-\tau]} \mC_k^\her\stackrel{\mathrm{(a)}}{=} R_k[\tau] \mC_k\mC_k^\her,
\end{align}
\end{subequations}
where (a) follows $\Ex{\bgg_k [t] \bgg_k^\her [t-\tau]}  =  R_k[\tau] \mI$. 
 \section{}
Let start from the definition, we have
\begin{align}
	\crb \brc{\{\mC_k\} \vert \bvv} =  \tr{\mJ^{-1}\brc{\bvv}} \stackrel{\mathrm{(a)}}{\geq} \frac{9}{\tr{\mJ\brc{\bvv}}},\label{eq:ineq}
\end{align}
where (a) follows from the fact that for any inevitable $n\times n$ matrix $\mA$, we have $\tr{\mA}\tr{\mA^{-1}} \geq n^2$. In the following, we use the following lemma.
Using Lemma~\ref{lem:3}, we can write
\begin{subequations}
\begin{align}
	&\tr{\mJ\brc{\bvv}} = \sum_{i=x,y,z} [\mJ\brc{\bvv}]_{i,i}, \\[-2pt]&\leq 
   E\brc{\{\mC_k\}} \sum_{k=1}^K \frac{T}{\varsigma_k^2} \sum_{i=x,y,z}
   \norm{ \frac{\partial}{\partial \ell_i} \baa_{k} \brc{\bold{v}} \baa_k \brc{\bold{v}}^\trp }_F^2.\label{eq:42}
	\end{align}
 \end{subequations}
The bound in \eqref{eq:42} is tight when optimized over~$\bvv$. \textcolor{black}{We note that $\baa_{k} \brc{\bold{v}}$ follows the classical array response model given in \eqref{array_resp}, with $\alpha_k = \frac{\alpha_0}{\|\mathbf{v} - \mathbf{v}_k\|_2}$ and $\theta_k(\mathbf{v}) = \arcsin\left(\frac{\ell_z - \ell_{kz}}{\|\mathbf{v} - \mathbf{v}_k\|_2} \right)$.  Since $\mathbf{v}$ lies in a compact region excluding the singularity at $\mathbf{v}=\mathbf{v}_k$, both $\theta_k(\mathbf{v})$ and $\alpha_k$ are smooth functions with bounded derivatives. Consequently, the entries of  $\alpha_k$ which are composed of trigonometric and rational functions of $\mathbf{v}$, are differentiable with bounded partial derivatives. This ensures that the gradients of $\baa_{k} \brc{\bold{v}}$ are bounded over the region of interest \cite{orfanidis1995introduction}}.
This implies that there exists $\epsilon > 0 $ such that 
\begin{align}
\sum_{i=x,y,z} \norm{ \frac{\partial}{\partial \ell_i} \baa_{k} \brc{\bold{v}} \baa_k \brc{\bold{v}}^\trp }_F^2 \leq \frac{1}{\epsilon}.\label{eq:BND}
\end{align}
Let $\epsilon = \varrho/9$. For any $\bvv$, we have $\crb \brc{\{\mC_k\} \vert \bvv} \geq \frac{\varrho \bar{\varsigma}^2}{T E\brc{\{\mC_k\}}}$ with $\bar{\varsigma}^2$ in Theorem 2. 
From \eqref{eq:worstcrb}$,\crb_{\text{worst}} \brc{\{\mC_k\}} \geq \crb \brc{\{\mC_k\} \vert \bvv}$ for all $\bvv$. 
Note that this universal bound holds for all target parameters $\bvv$, with its tightness governed by $\epsilon$. Choosing $\epsilon$ to make \eqref{eq:BND} tight at the worst-case $\bvv$ ensures Theorem~2’s bound is tight. Since  $\epsilon$ affects the bound only via the constant $\varrho$, minimizing the bound in Theorem~2 effectively minimizes the worst-case CRB, independent of $\varrho$.

\bibliographystyle{IEEEtran} 
\bibliography{ref}  

\end{document}